\documentclass{llncs}

\usepackage{booktabs}   %% For formal tables:
\usepackage{ulem} %% strike through

\usepackage{inconsolata} % bold monospaced font

\usepackage{listings}
\usepackage{url}

\usepackage{amsmath}
\usepackage{amssymb}
\usepackage{stmaryrd}

\usepackage{tabularx}
\usepackage{mathpartir}
\usepackage{ifthen}
\usepackage{balance}
\usepackage{cleveref}

\usepackage{color}
\definecolor{GrayBgColor}{rgb}{0.9, 0.9, 0.9}
\definecolor{GrayFgColor}{rgb}{0.4, 0.4, 0.4}
\definecolor{StringColor}{rgb}{0.0, 0.38039, 0.141176}
\usepackage{graphicx}

\renewcommand\emph[1]{{\it #1}}

\newcommand{\mathem}{\sf}
\newcommand{\kw}[1]{\normalfont\bfseries\sffamily #1}
\newcommand\tkw[1]{{\bfseries #1}}
\newcommand{\IN}{\mbox{\kw{in}}}
\newcommand{\LET}{\mbox{\kw{let}}}

\newcommand{\CASE}{\mbox{\kw{case}}}

\newcommand{\OF}{\mbox{\kw{of}}}

\newcommand{\foreach}[2]{\overline{#1}^{#2}} %% {{\overrightarrow{#1}}^{#2}}
\newcommand{\foreachN}[1]{\overline{#1}}    %% index and bound do not matter

\newcommand{\gap}{\,\,\,\,}

\newcommand{\TYPE}{\mbox{\kw{type}}}
\newcommand{\STRUCT}{\mbox{\kw{struct}}}
\newcommand{\INTERFACE}{\mbox{\kw{interface}}}
\newcommand{\FUNC}{\mbox{\kw{func}}}
\newcommand{\RETURN}{\mbox{\kw{return}}}
\newcommand{\PACKAGE}{\mbox{\kw{package}}}
\newcommand{\MAIN}{\mbox{\mathem main}}

\newcommand\GoSynCatName[1]{\mbox{\small #1}}

\newcommand{\methodSpecificationsRel}{{\mathem methods}}
\newcommand{\methodSpecifications}[2]{\methodSpecificationsRel(#1,#2)}

\newcommand{\fgEnv}{\Gamma}
\newcommand{\EmptyFgEnv}{\{\}}
\newcommand{\turnsFG}{\, \vdash_{\mathsf{FG}} \,}
\newcommand{\subtypeOf}[2]{#1 <: #2}
\newcommand\assertOfSym{\mathrel{\scalebox{0.8}{\ensuremath{\searrow}}}}
\newcommand{\assertOf}[2]{#1 \assertOfSym #2}

\newcommand{\fgSub}[3]{#1 \turnsFG \subtypeOf{#2}{#3}}

\newcommand{\FGStructNonRec}{{FG1}}
\newcommand{\FGUniqueFields}{{FG2}} %% {{\bf FG-UNIQUE-FIELD}}
\newcommand{\FGUniqueMethSpec}{{FG3}} %% {{\bf FG-UNIQUE-MSPEC}}
\newcommand{\FGUniqueReceiver}{{FG4}} %% {{\bf FG-UNIQUE-RECV}}

\newcommand{\panicFG}[2]{\mathsf{panic}_{\mathsf{FG}}(#1, #2)}
\newcommand{\panicDFG}[1]{\panicFG{\foreachN{D}}{#1}}

\newcommand{\divergeFG}[2]{#1 \Uparrow_{\mathsf{FG}} #2}
\newcommand{\divergeDFG}[1]{\divergeFG{\foreachN{D}}{#1}}

\newcommand{\EvCtx}{{\mathcal E}}
\newcommand{\reduceSym}{\longrightarrow}
\newcommand{\reduce}[2]{#1 \reduceSym  #2}
\newcommand{\reduceN}[2]{#1 \reduceSym^* #2} %% reduce to normal form
\newcommand{\reducek}[3]{#2 \reduceSym^{#1} #3} %% reduce k steps
\newcommand{\subst}[2]{#1 \mapsto #2}

\newcommand{\reduceFGN}[3]{#1 \turnsFG #2 \reduceSym^* #3} %% reduce to normal form
\newcommand{\reduceFGk}[4]{#2 \turnsFG #3 \reduceSym^{#1} #4} %% reduce to normal form

\newcommand{\reduceTL}[3]{#1 \turnsTL #2 \reduceSym #3}
\newcommand{\reduceTLN}[3]{#1 \turnsTL #2 \reduceSym^* #3} %% reduce to normal form
\newcommand{\reduceTLk}[4]{#2 \turnsTL #3 \reduceSym^{#1} #4} %% reduce k steps

\newcommand{\panic}[1]{{\mathem panic}(#1)}
\newcommand{\diverge}[1]{{\mathem diverge}(#1)}

\newcommand{\vbName}{\Phi}

\newcommand{\vbFG}{\vbName_{\mathsf{v}}} %% Name of value binding

\makeatletter
\newcommand{\xRightarrow}[2][]{\ext@arrow 0359\Rightarrowfill@{#1}{#2}}
\makeatother

\newcommand{\panicMTL}[1]{\mathsf{panic}_{\mathsf{TL}}(\vbMethTL, #1)}

\newcommand{\divergeMTL}[1]{\vbMethTL \Uparrow_{\mathsf{TL}} #1}

\newcommand{\TL}{\mbox{TL}} %% Target language = GADT where for convenience we add ``method'' lookup to reduction rules

\newcommand{\kT}{K}         %% constructor
\newcommand\Angle[1]{\langle#1\rangle}
\newcommand{\pair}[2]{\Angle{#1,#2}}
\newcommand{\triple}[3]{\Angle{#1,#2, #3}}

\newcommand{\program}{\mathit{Prog}}

\newcommand{\clsT}{\mathit{Cls}}
\newcommand{\patT}{\mathit{Pat}}
\newcommand{\expT}{E}
\newcommand{\expTA}{\expT'}    %% N.B. Can't use \expT_1

\newcommand{\xT}{X}
\newcommand{\yT}{Y}

\newcommand{\xTval}{\xT_{\mathit{val}}}

\newcommand{\kA}{k'}
\newcommand{\kB}{k''}
\newcommand{\kC}{k'''}
\newcommand{\kD}{k''''}

\newcommand{\minF}{\mathit min}
\newcommand{\maxF}{\mathit max}

\newcommand{\uT}{V}                 %% TL value, use V
\newcommand{\uTA}{\uT'}

\newcommand{\uTVal}{\uT_{val}}

\newcommand{\ignore}[1]{}

\newcommand{\PANIC}[1]{#1}

\newcommand{\addSpace}{\mbox{} \\}

\newcommand{\reduceSourcek}[3]{#2 \reduceSym^{#1} #3}
\newcommand{\reduceTargetN}[2]{#1 \reduceSym^* #2}
\newcommand{\reduceTargetk}[3]{#2 \reduceSym^{#1} #3}

\newcommand{\txt}[1]{\texttt{#1}}

\newcommand{\aplasLR}[4]{#3 \approx_{\mathit{APLAS}} #4 \in \Sem{#2}_{#1}}
\newcommand{\aplasVLR}[4]{#3 \equiv_{\mathit{APLAS}} #4 \in \Sem{#2}_{#1}}

\newcommand{\terminatingLR}[4]{#3 \approx #4 \in \Sem{#2}_{#1}}
\newcommand{\divergeLR}[4]{#3 \approx_{\uparrow} #4 \in \Sem{#2}_{#1}}
\newcommand{\panicLR}[4]{#3 \approx_{\not \downarrow} #4 \in \Sem{#2}_{#1}}
\newcommand{\valueLR}[4]{#3 \equiv #4 \in \Sem{#2}_{#1}}
\newcommand{\expLR}[4]{#3 \approx #4 \in \Sem{#2}_{#1}}

\newcommand{\tcaseof}[2]{\CASE\ #1 \ \OF\ #2}

\newcommand{\tletrecin}[2]{\LET\ #1 \ \IN\ #2}  %% MS: just use \LET instead of \LETREC

\newcommand\turnsSomething[1]{\, \vdash_{\mathsf{#1}} \,}
\newcommand\turnsSomethingk[2]{\, \vdash^{#2}_{\mathsf{#1}} \,}

\newcommand{\turnsTL}{\turnsSomething{TL}}

\newcommand{\vbTL}{\vbName_{\mathsf{\uT}}}
\newcommand{\vbMethTL}{\vbName_{\mathsf m}}

\newcommand{\REvCtxT}{{\mathcal R}} %% context for the refined case

\newcommand{\turnsG}[1]{\, \vdash_{#1} \,}

\newcommand{\trans}[4]{#2 \turnsG{#1} #3 \leadsto #4}
\newcommand{\transGrayBox}[4]{#2 \turnsG{#1} #3 \GrayBox{\leadsto #4}}

\newcommand{\tdExpTrans}[3]{\trans{\mathsf{exp}}{#1}{#2}{#3}}
\newcommand{\tdExpTransGrayBox}[3]{\transGrayBox{\mathsf{exp}}{#1}{#2}{#3}}
\newcommand{\tdMethTrans}[3]{\trans{\mathsf{meth}}{#1}{#2}{#3}}
\newcommand{\tdMethTransGrayBox}[3]{\transGrayBox{\mathsf{meth}}{#1}{#2}{#3}}
\newcommand{\tdProgTrans}[2]{ \turnsG{\mathsf{prog}} #1 \leadsto #2}
\newcommand{\tdProgTransGrayBox}[2]{ \turnsG{\mathsf{prog}} #1 \GrayBox{\leadsto #2}}

\newcommand{\tdUpcast}[3]{\trans{\mathsf{iCons}}{#1}{#2}{#3}}
\newcommand{\tdUpcastGrayBox}[3]{\transGrayBox{\mathsf{iCons}}{#1}{#2}{#3}}
\newcommand{\tdDowncast}[3]{\trans{\mathsf{iDestr}}{#1}{#2}{#3}}
\newcommand{\tdDowncastGrayBox}[3]{\transGrayBox{\mathsf{iDestr}}{#1}{#2}{#3}}

\newcommand\methodLookupSym{{\mathem methodLookup}}
\newcommand{\methodLookup}[2]{\methodLookupSym(#1, #2)}

\newcommand{\mapPerm}{\pi}               %% total function

\newcommand{\redLRk}[5]{#2 \turnsSomethingk{rr}{#1} #4 \approx_{#3} #5}

\newcommand\Sem[1]{\llbracket#1\rrbracket}
\newcommand{\redLRConfk}[4]{ #3 \approx #4 \in \Sem{#2}_{#1}^{\pair{\foreachN{D}}{\vbMethTL}}}

\newcommand{\redVLRConfk}[4]{ #3 \equiv #4 \in \Sem{#2}_{#1}^{\pair{\foreachN{D}}{\vbMethTL}}}

\newcommand{\mT}[2]{\xT_{{#1},{#2}}}

\newcommand\sSynCatName[1]{\mbox{\small #1}}

\newcommand\sTransSection[2]{%
  \begin{tabularx}{\textwidth}{Xr}%
    #1\hfill & {\small \it #2}%
  \end{tabularx}%
}

\newcommand\sOver[2][!*NEVER USED ARGUMENT*!]{%
  \ifthenelse{\equal{#1}{!*NEVER USED ARGUMENT*!}}{\overline{#2}}{\overline{#2}^{#1}}%
}

\newcommand\sOptEmpty\bullet

\newcommand\RefTirName[1]{\TirName{#1}}
\newcommand\Rule[1]{\RefTirName{#1}}

\newcommand\inferruleLeft[3][]{\inferrule*[vcenter, Left=#1]{#2}{#3}}

\newcommand\GrayBox[1]{#1}

\newcommand{\bi}{\begin{array}[t]{@{}l@{}}}
\newcommand{\ei}{\end{array}}
\newcommand{\ba}{\begin{array}}
\newcommand{\ea}{\end{array}}
\newcommand{\bda}{\[\ba}
\newcommand{\eda}{\ea\]}
\newcommand{\bp}{\begin{quote}\tt\begin{tabbing}}
\newcommand{\ep}{\end{tabbing}\end{quote}}

\def\ruleform#1{{\setlength{\fboxrule}{0.5pt}\fbox{\normalsize \ensuremath{#1}}}}

\newcommand{\myirule}[2]{{\renewcommand{\arraystretch}{1.2}\ba{c} #1
                      \\ \hline #2 \ea}}

\newcommand{\figurebox}[1]
        {\fbox{\begin{minipage}{\textwidth} #1 \end{minipage}}}
\newcommand{\boxfig}[3]
           {\begin{figure*}\figurebox{#3}\vspace{-2ex}\caption{\label{#1}#2}\end{figure*}}

\newcommand\noteForall[1]{(\forall\, #1)}

\newcommand\Nat{\mathbb{N}}
\newenvironment{EnumAlph}{%

  \begin{enumerate}
}{\end{enumerate}}

\makeatletter
\def\smallunderbrace#1{\mathop{\vtop{\m@th\ialign{##\crcr
   $\hfil\displaystyle{#1}\hfil$\crcr
   \noalign{\kern3\p@\nointerlineskip}%
   \tiny\upbracefill\crcr\noalign{\kern3\p@}}}}\limits}
\makeatother
\newcommand\BraceBelow[2]{\smallunderbrace{#1}_{#2}}
\newcommand\Hole{\scalebox{0.8}{\ensuremath{\square}}}

\begin{document}

%% Title information
\title{Semantic preservation for a type directed translation scheme of Featherweight Go}

\author{Martin Sulzmann\inst1  \and Stefan Wehr\inst2}
\institute{
  Karlsruhe University of Applied Sciences, Germany\\
  \email{martin.sulzmann@h-ka.de}
  \and
  Offenburg University of Applied Sciences, Germany\\
  \email{stefan.wehr@hs-offenburg.de}
}

\maketitle

\begin{abstract}
  Featherweight Go (FG) is a minimal core calculus that includes
  essential Go features such as overloaded methods and interface types.
  The most straightforward semantic description of the dynamic behavior of FG programs
  is to resolve method calls based on run-time type information.
  A more efficient approach is to apply a type-directed translation scheme
  where interface-values are replaced by dictionaries that contain concrete method definitions.
  Thus, method calls can be resolved by a simple lookup of the method definition in the dictionary.
  Establishing that the target program obtained via the type-directed translation scheme
  preserves the semantics of the original FG program is an important task.

  To establish this property we employ logical relations that are indexed by types to relate source and target programs.
  We provide rigorous proofs and give a detailed discussion of the many subtle corners that we have
  encountered including the need for a step index due to recursive interfaces and method definitions.
  %% Methods developed here should be applicable to Haskell, Rust, Scala, ...
  %% Maybe leave out and mention somewhere in the conclusion.
\end{abstract}

%--------------------------------------------------------
%--------------------------------------------------------
\section{Introduction}

Type directed translation is the process of elaborating a source into some target program
by making use of type information available in the source program.
Source and target may be from the same language as found in the case of compiler transformations,
e.g.~consider~\cite{Morrisett1995CompilingWT}.
The target may be a more elementary language compared to the source,
e.g.~consider~\cite{Hall:1996:TCH:227699.227700,10.1145/1291151.1291169}.
In all cases it is essential to establish that the target program resulting from
the translation preserves the meaning of the source program.

Here, we consider a type directed translation method applied to Featherweight Go.
Featherweight Go (FG) is a minimal core calculus that includes the essential features of Go
such as method overloading and interfaces.
Earlier work by Griesmer and co-authors~\cite{FeatherweightGo}
specifies static typing rules and a run-time method lookup semantics for FG.
In our own prior work~\cite{DBLP:conf/aplas/SulzmannW21},
we give a type directed translation that
elaborates methods calls to method lookup in dictionaries that will be passed
around in place of interfaces.
We could establish correctness of our translation but
the result was somewhat limited as semantic preservation only holds
under the assumption that source and target programs terminate.

In this work, we significantly extend our earlier semantic preservation result
and establish the following properties.
\begin{itemize}
  \item If the source program terminates so will the target program and the resulting values are equivalent.
  \item If the source program diverges so will the target program.
  \item If the source program panics due to a failed run-time type check, the target program will panic as well.
\end{itemize}
These results require non-trivial extensions and adaptations
of our earlier proof method and type-indexed relation to
connect source to target values.
The upcoming Section~\ref{sec:overview} gives an overview
and highlights the changes from~\cite{DBLP:conf/aplas/SulzmannW21}
to achieve the above results.

In summary, we make the following contributions.

\begin{itemize}
\item We introduce a family of syntactic, step-indexed logical relations to establish semantic
  preservation for terminating, diverging and panicking source programs
  (Section~\ref{sec:logical-relations})
\item We provide for rigorous proofs of our results
       (Section~\ref{sec:properties} and Appendix).
\end{itemize}

Section~\ref{sec:featherweight-go} specifies Featherweight Go (FG).
The type directed translation of FG including a description
of the target language is given in Section~\ref{sec:type-directed}.
Both sections are adopted from our earlier work~\cite{DBLP:conf/aplas/SulzmannW21}.
Related work is covered in Section~\ref{sec:related}.
Section~\ref{sec:conc} concludes.

%--------------------------------------------------------
%--------------------------------------------------------
\section{Overview}
\label{sec:overview}

\newcommand{\Eq}{\txt{Eq}}
\newcommand{\Int}{\txt{Int}}
\newcommand{\eqM}{\txt{eq}}
\newcommand{\stepCmp}[2]{#1 \blacktriangleleft #2}

\paragraph{Translation by example.}

We consider a type directed translation scheme that
transforms a FG program into some target program.
The target language
is an untyped lambda-calculus extended with recursive let-bindings, constructors, and pattern matching.
Here, we use Haskell-style notation.

For example, the FG program on the left translates to the target program on the right.
For simplicity, we leave out some details (marked by \texttt{...}).

\mbox{}
\\
{\small
\begin{tabular}{l@{\qquad}l}
\texttt{\tkw{type} Int \tkw{struct} \{val int\}}
\\ \texttt{\tkw{type} Eq \tkw{interface} \{eq(that Eq) bool\}}
\\ \texttt{\tkw{func} (this Int) eq(that Eq) bool \{...\}}     & \ \ \ \texttt{eqInt this that = ...}
\\ \texttt{\tkw{func} main() \{}                               &  \ \ \ \texttt{main =}
\\ \ \ \texttt{\tkw{var} i Eq = Int\{1\}}          & \quad \ \ \texttt{\tkw{let} j = (1,eqInt)}
\\ \ \ \texttt{\tkw{var} \_ bool = i.eq(i) \}}        &  \quad \ \ \texttt{\tkw{in} \tkw{case} j \tkw{of} (x,eq) -> eq x j}
\end{tabular}
}
\mbox{}
\\

The FG program on the left contains a struct \txt{Int}, an interface \txt{Eq}, and
a definition for method \txt{eq} for receiver type \txt{Int} (line 3).
Our example only contains one definition of the \txt{eq} method.
In general, FG methods can be overloaded on the receiver type.
Hence, in the translation on the right, the function name \texttt{eqInt} uniquely identifies
the definition of \texttt{eq} for receiver type \texttt{Int}.

Interfaces give a name to a set of method signatures.
They are types, and so interface type \texttt{Eq} effectively describes all receiver types
implementing the \texttt{eq} method.
Type \texttt{Int} implements this \texttt{eq} method
and therefore \texttt{Int\{1\}} is (also) of type \texttt{Eq}.
Hence, the method call \texttt{i.eq(i)} type checks.

The FG semantics performs a run-time type lookup to resolve method calls such as \texttt{i.eq(i)}.
In the translation, an interface is represented as a pair that consists
of the value that implements the interface and a dictionary of method definitions
for this specific value.
For example, \texttt{i} at type \Eq\ translates to the pair \texttt{(1,eqInt)}.
Assuming that we represent the FG variable \txt{i} as \txt{j} in the target,
the method call \texttt{i.eq(i)} translates to \texttt{\tkw{case} j \tkw{of} (x,eq) -> eq x j},
where we only require a pattern match
to access the underlying value and the concrete method definition.
See \Cref{sec:type-directed} for details of the type-directed translation.

\paragraph{Semantic preservation via logical relations.}

To establish that the translation is meaning preserving we need to relate source to target expressions.
One challenge is that evaluation steps are not in sync.
For example, FG method calls reduce in a single step whereas the translated code
first performs a pattern match to obtain the method definition
followed by another step to execute the call.

\boxfig{f:highlights}{Improvements compared to earlier work}{

  Earlier work:
\begin{mathpar}
  \inferrule
            { \forall k_1, k_2, v, \uT ~.~
              k - (k_1 + k_2) > 0
              \ \wedge \
              \reduceSourcek{k_1}{e}{v}
              \ \wedge \
              \reduceTargetk{k_2}{\expT}{\uT}
              \\
              \implies \
              \aplasVLR{k-(k_1 + k_2)}{t}{v}{\uT}
            }
            {
              \aplasLR{k}{t}{e}{\expT}
            }
\end{mathpar}

This work:
\begin{mathpar}
  \inferrule[terminate]
  {
       \forall k' < k,
       v ~.~ \reduceSourcek{\kA}{e}{v} \implies \exists \uT.
       \reduceTargetN{\expT}{\uT} \wedge \valueLR{k- \kA}{t}{v}{\uT}
  }
  {
    \terminatingLR{k}{t}{e}{\expT}
  }

  \inferrule[diverge]
  {
       \forall k' < k,
       e' ~.~ \reduceSourcek{\kA}{e}{e'} \wedge \diverge{e'}
       \implies
       \diverge{\expT}
  }
  {
    \divergeLR{k}{t}{e}{\expT}
  }

  \inferrule[panic]
  {
       \forall k' < k,
       e' ~.~ \reduceSourcek{\kA}{e}{e'} \wedge \panic{e'}
       \implies
       \panic{\expT}
  }
  {
    \panicLR{k}{t}{e}{\expT}
  }

\end{mathpar}

} %% fig

In our earlier work~\cite{DBLP:conf/aplas/SulzmannW21}, we introduce
a logical relation $\aplasLR{k}{t}{e}{\expT}$ to express that source and target expressions behave the same.
See top of Figure~\ref{f:highlights}.
The relation is indexed by type $t$ and a step index $k$
where we assume that step indices are natural numbers starting with $0$.
Source expression $e$ and target expression $\expT$ are in a relation:
if the sum of evaluation steps to reduce $e$ to some source value $v$
and $\expT$ to some target value $\uT$ is less than $k$, then the values must be in a relation.\footnote{%
  We defer all formal and missing definitions to Sections~\ref{sec:featherweight-go}, \ref{sec:type-directed}
  and~\ref{sec:semantic-preservation}. For now, we appeal to intuition.
}
Thus, the number of reduction steps in the source and target do not need to be in sync.
If we need more than $k$ steps, or if only the source expression yields a value within $k$ steps,
or if one of the expressions diverges or panics,
the relation $\aplasLR{k}{t}{e}{\expT}$ holds vacuously and does not give us any information.

In this work, we consider three cases: source terminates, diverges or panics.
See Figure~\ref{f:highlights} that sketches a logical relation for each case.
\Rule{terminate}:
if the source $e$ terminates within less than $k$ steps,
then the target $\expT$ terminates as well where the number of evaluation steps do not matter.
\Rule{diverge} / \Rule{panic}:
if the source $e$ evaluates to $e'$ in less than $k$ steps and $e'$ diverges/panics, then so does target $\expT$.
The detour via $e'$ in the last two cases is required to prove that source $e$ and its translation $E$ are
related.
Taken together, the three cases yield a much stronger characterization of the semantic relation
between source and target programs compared to our earlier work.

We use the convention that $\equiv$ relates values whereas $\approx$ relates
expressions.
The step index in the relation for values, see $\valueLR{k- \kA}{t}{v}{\uT}$, seems unnecessary
but is important to guarantee that the definition of logical relations is well-founded
(c.f.~\Cref{sec:semantic-preservation})

\paragraph{Ill-founded without step index.}

Consider the example from above where
\txt{Int\{1\}} of type \txt{Eq} translates to \txt{(1,eqInt)}.
We expect $\valueLR{}{\txt{Eq}}{\txt{Int\{1\}}}{\txt{(1,eqInt)}}$ to hold; that is,
FG value $\txt{Int\{1\}}$ is equivalent to target value $\txt{(1,eqInt)}$ when viewed
at FG type $\txt{Eq}$.

The following reasoning steps try to verify this claim.
We deliberately ignore the step index to show that without a step index we run into some issue.
%
%{\small
\bda{rl}
& (1) \ \valueLR{}{\txt{Eq}}{\txt{Int\{1\}}}{\txt{(1,eqInt)}}
\\ [\smallskipamount] \mbox{if} & (2) \ \valueLR{}{\txt{Int}}{\txt{Int\{1\}}}{\txt{1}}
\\ [\smallskipamount] \mbox{and}&    (3) \ \expLR{}{\txt{eq(y Eq) bool}}{\texttt{\tkw{func} (x Int) eq(y Eq) bool \{e'\}}}{\txt{eqInt}}
\eda
%}%
%
Statement (1) reduces to statements (2) and (3).
(2) states that the underlying values are related at struct type $\txt{Int}$. This clearly holds, we omit the details.
(3) requires a bit more thought.
Function \txt{eqInt} is part of the dictionary. Hence, \txt{eqInt} must have the same behavior
as method \txt{eq} defined on receiver type \txt{Int}.
This is the intention of statement (3).
Compared to the earlier notation of the example, the
function body has been replaced by some expression \txt{e'}.

How can we establish (3)?
It must hold that when applied to related argument values,
\txt{eqInt} and method \txt{eq} defined on receiver type \txt{Int} behave the same.
Thus, establishing (3) requires (4):
%{\small
\bda{ll}
   & (4) \ \forall \valueLR{}{\txt{Int}}{v}{\uT}, \valueLR{}{\txt{Eq}}{v'}{\uTA}.
 \  \expLR{}{\txt{bool}}{\Angle{\subst{\txt{x}}{v},\subst{\txt{y}}{v'}} \txt{e'}}{\txt{eqInt} \ \uT \ \uTA}
\eda
%}%
%
where we write $\Angle{\subst{\txt{x}}{v},\subst{\txt{y}}{v'}} \txt{e'}$
to denote the substitution of arguments by values in the function body.

There is an issue. One of the arguments is of interface type {\Eq}.
Hence, for any values $v', \uTA$ we require $\valueLR{}{\txt{Eq}}{v'}{\uTA}$.
This leads to a cyclic dependency as
a statement of the form $\valueLR{}{\txt{Eq}}{\cdot}{\cdot}$
relies on a statement $\valueLR{}{\txt{Eq}}{\cdot}{\cdot}$.
See (1) and (4).
This would mean that the definition of our logical relation is ill-founded.

\paragraph{Step indices to the rescue.}

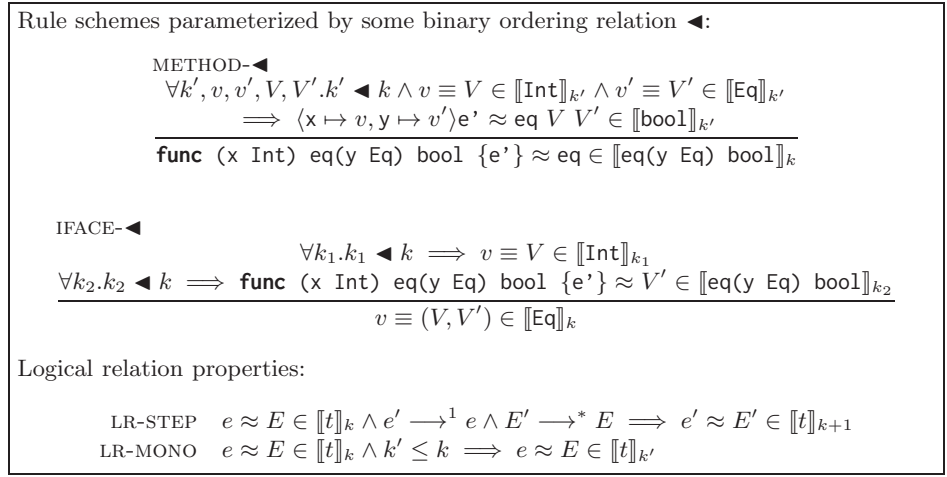
\boxfig{f:step-index}{Getting the step index right}{

 Rule schemes parameterized by some binary ordering relation $\stepCmp{}{}$:
  \begin{mathpar}
    \inferrule[method-$\stepCmp{}{}$]
              { \forall \kA, v, v', \uT, \uTA.
                \stepCmp{\kA}{k}
                \wedge
                \valueLR{\kA}{\Int}{v}{\uT}
                \wedge
                \valueLR{\kA}{\Eq}{v'}{\uTA}
                \\
                \implies  \expLR{\kA}{\txt{bool}}{\Angle{\subst{\txt{x}}{v},\subst{\txt{y}}{v'}} \txt{e'}}{\txt{eq} \ \uT \ \uTA}              }
              {
                \expLR{k}{\txt{eq(y Eq) bool}}{\texttt{\tkw{func} (x Int) eq(y Eq) bool \{e'\}}}{\eqM}
              }
  \end{mathpar}

  \begin{mathpar}
    \inferrule[iface-$\stepCmp{}{}$]
              { \forall k_1. \stepCmp{k_1}{k} \implies \valueLR{k_1}{\Int}{v}{\uT}
                \\ \forall k_2. \stepCmp{k_2}{k} \implies \expLR{k_2}{\txt{eq(y Eq) bool}}{\texttt{\tkw{func} (x Int) eq(y Eq) bool \{e'\}}}{V'}
              }
              { \valueLR{k}{\Eq}{v}{(\uT,V')}
              }
  \end{mathpar}

  Logical relation properties:

\bda{r@{\quad}l}
\Rule{lr-step} & \expLR{k}{t}{e}{\expT} \wedge \reduceSourcek{1}{e'}{e} \wedge \reduceTargetN{E'}{E}
\implies
\expLR{k+1}{t}{e'}{E'}
\\
\Rule{lr-mono} & \expLR{k}{t}{e}{\expT} \wedge k' \leq k \implies \expLR{k'}{t}{e}{\expT}
\eda

}

FG interfaces can have cyclic dependencies similar to recursive types, see interface \Eq.
To guarantee well-foundedness we include a step index in the definition of our logical relations.
There is in fact a second reason for a step index. Method definitions may be recursive similar to recursive functions.
There is no well-foundedness issue here.
But to apply an inductive proof method where semantic preservation for expressions is lifted
to method definitions require a step-index.

Step indices in case of recursive types and recursive functions have been studied before~\cite{10.1007/11693024_6}.
What makes our setting interesting is a subtle interaction between (recursive) interfaces and (recursive) methods.
Figure~\ref{f:step-index} specifies the logical relation rules for methods and interfaces
that we have used in the above reasoning steps (1-4).
The rule for interfaces relies on the rule for methods
and the rule for methods relies on the rule for interfaces (in case of recursive interfaces).
For brevity, we omit rules for struct types such as \Int.

Rules \Rule{method-$\stepCmp{}{}$} and \Rule{iface-$\stepCmp{}{}$}
are parameterized by some binary ordering relation $\stepCmp{}{}$.
Why not simply replace $\stepCmp{}{}$ by $<$, the less than relation?
Rule instances \Rule{method-$<$} and \Rule{iface-$<$} are clearly well-founded.

\paragraph{Failed proof attempt in case of \Rule{method-$<$} and \Rule{iface-$<$}.}

Via our running example we illustrate that
the proof of semantic preservation for expressions will not go through.
Recall that
\bda{l@{\quad}c@{\quad}l}
\txt{i.eq(i)} \textrm{~of type~} \txt{bool} & \textit{translates to} & \txt{\tkw{case} j \tkw{of} (x,eq) -> eq x j}\\
\txt{i = Int\{1\}} \textrm{~of type~} \txt{Eq} & & \txt{j = (1, eqInt)}
\eda
For values \txt{i} and \txt{j}, we may assume (1) $\valueLR{k}{\Eq}{\txt{Int\{1\}}}{\txt{(1,eqInt)}}$
for some $k$.
To verify that the translation yields related expressions,
we must show
$$(2)~ \expLR{k}{\txt{bool}}{\txt{i.eq(i)}}{\txt{\tkw{case} j \tkw{of} (x,eq) -> eq x j}}$$
From (1), via reverse application of rule \Rule{iface-$<$}, we can derive
$$(3)~\expLR{k-1}{\txt{eq(y Eq) bool}}{\texttt{\tkw{func} (x Int) eq(y Eq) bool \{e'\}}}{\txt{eqInt}}$$

From (3) we get the implication in the premise of rule \Rule{method-$<$}.
The left-hand side of this implication can be satisfied for $k-2 < k -1$
via \Rule{lr-mono} from \Cref{f:step-index} and (1) and
the fact that $\valueLR{k}{\txt{Int}}{\txt{Int\{1\}}}{1}$.
Thus, we can derive the right-hand side
$\expLR{k-2}{\txt{bool}}{\Angle{\subst{\txt{x}}{\txt{i}}, \subst{\txt{y}}{\txt{i}}} \txt{e'}}{\txt{eqInt 1 j}}$.
From this we get
$$(5)~\expLR{k-1}{\txt{bool}}{\txt{i.eq(i)}}{\txt{\tkw{case} j \tkw{of} (x,eq) -> eq x j}}$$
via property \Rule{lr-step} in Figure~\ref{f:step-index} and
the following evaluation steps:
\bda{rcl}
\txt{i.eq(i)}                                & \reduceSym^1 & \Angle{\subst{\txt{x}}{\txt{i}}, \subst{\txt{y}}{\txt{i}}} \txt{e'}\\
\txt{\tkw{case} j \tkw{of} (x,eq) -> eq x j} & \reduceSym^* & \txt{eqInt 1 j}\\
\eda

The issue is that from (5) we cannot deduce (2).
Property \Rule{lr-step} allows us to bump up the step index in case of source reduction steps.
Target reductions have no impact.
Hence, we end up being one step short.

\paragraph{Fixing the proof by turning $<$ into $\leq$.}

The solution is to turn one $<$ into $\leq$.
Then, we can derive (5) $\expLR{k}{\txt{bool}}{\txt{i.eq(i)}}{\txt{\tkw{case} j \tkw{of} (x,eq) -> eq x j}}$
and the proof of semantic equivalence for expressions goes through.
It seems that we have a choice between
(A) rule instances \Rule{method-$\leq$} and \Rule{iface-$<$} and
(B) rule instances \Rule{method-$<$} and \Rule{iface-$\leq$}.
We pick choice (B) because under (A) the proof of semantic preservation for (possibly recursive) method definitions
will not go through. See the proof of the upcoming Lemma~\ref{le:method-red-rel-equiv}
in \Cref{sec:semantic-preservation}.

\paragraph{Comparison to our earlier work~\cite{DBLP:conf/aplas/SulzmannW21}.}

The logical relation introduced in our earlier work~\cite{DBLP:conf/aplas/SulzmannW21}
is more limited in that semantic preservation is only stated
under the assumption that both expressions, source and target programs, terminate.
Recall Figure~\ref{f:highlights} that shows
that the logical relation $\aplasLR{}{}{}{}$ also takes into account source as well as target steps.
Under this stronger assumption it is easier to get the step index right as we can derive
the following more general variant of property \Rule{lr-step}
\bda{lc}
 \aplasLR{k}{t}{e}{\expT} \wedge \reduceSourcek{k_1}{e'}{e} \wedge \reduceTargetk{k_2}{\expT}{\expTA}
\implies
\aplasLR{k+k_1+k_2}{t}{e'}{\expTA}
\eda
That is, we can bump up the step index based on target reductions as well.
This provides for more flexibility, even with
rule instances \Rule{method-$<$} and \Rule{iface-$<$},
the proofs go through.
As highlighted above, more care is needed for the logical relations that we introduce in this work.

Next, we introduce the semantics of Featherweight Go
and give the details of the typed-directed translation scheme
followed by our semantic preservation results.

%--------------------------------------------------------
%--------------------------------------------------------
\section{Featherweight Go}
\label{sec:featherweight-go}

Featherweight Go (FG)~\cite{FeatherweightGo} is a tiny fragment of Go containing
only structs, methods and interfaces.
Figure~\ref{f:fg} gives its syntax and the dynamic semantics.
Overbar notation $\sOver[n]{\mathfrak s}$ denotes the sequence $\mathfrak s_1 \ldots \mathfrak s_n$ for some
syntactic construct $\mathfrak s$, where in some places commas separate the sequence
items.
If irrelevant, we omit the $n$ and simply write
$\sOver{\mathfrak s}$. Using the index variable $i$ under an overbar
marks the parts that vary from sequence item to sequence
item; for example, $\sOver[n]{\mathfrak s'\,\mathfrak s_i}$ abbreviates
$\mathfrak s'\,\mathfrak s_1\ldots\mathfrak s'\,\mathfrak s_n$
and $\sOver[q]{\mathfrak s_j}$ abbreviates
$\mathfrak s_{j1}\,\ldots\,\mathfrak s_{jq}$.

\boxfig{f:fg}{Featherweight Go (FG)}{
\bda{ll}
{
  \ba{ll}
  \GoSynCatName{Field name}     & f
\\ \GoSynCatName{Method name}   & m
\\ \GoSynCatName{Variable name} & x,y
\\ \GoSynCatName{Struct type name}    & t_S, u_S
\\ \GoSynCatName{Interface type name}    & t_I, u_I
\\ \GoSynCatName{Type name}              & t, u ::= t_S \mid t_I
\\ \GoSynCatName{Method signature}       & M ::= (\overline{x_i \ t_i}) \ t
\\ \GoSynCatName{Method specification}   & R, S ::= m M
\ea
}

&

\ba{llc}
\GoSynCatName{Expression} & d,e ::=
\\ \quad \GoSynCatName{Variable}          & \quad x & \mid
\\ \quad \GoSynCatName{Method call}       & \quad e.m(\overline{e}) & \mid
\\ \quad \GoSynCatName{Struct literal} & \quad t_S \{ \overline{e} \} & \mid
\\ \quad \GoSynCatName{Select}            & \quad e.f & \mid
\\ \quad \GoSynCatName{Type assertion}    & \quad e.(t)
\ea

\eda
\vspace{-1ex}
\bda{lcl}

\ba{llc}
\GoSynCatName{Type literal} & L ::=
\\ \quad \GoSynCatName{Struct} & \quad \STRUCT\ \{ \overline{f \ t} \} & \mid
\\ \quad \GoSynCatName{Interface} & \quad \INTERFACE\ \{ \overline{S} \}
\ea

&
\gap
&

\ba{llc}
\GoSynCatName{Declaration} & D ::=
\\ \quad \GoSynCatName{Type}   & \quad \TYPE\ t \ L & \mid
\\ \quad \GoSynCatName{Method} & \quad \FUNC\ (x \ t_S) \ mM \ \{ \RETURN\ e \}
\ea
\eda
\vspace{-1ex}
\bda{llcl}
\GoSynCatName{Program} & P & ::= & %% MS: omit \PACKAGE\ \MAIN;
          \overline{D} \ \FUNC\ \MAIN () \{ \_ = e \}
\eda

\sTransSection{
  \ruleform{\fgSub{\foreachN{D}}{t}{u}}
}{Subtyping}
  \begin{mathpar}
    \inferrule[methods-struct]
  {}
  {
    \methodSpecifications{\foreachN{D}}{t_S} = \{ m M \mid \FUNC\ (x \ t_S) \ mM \ \{ \RETURN\ e \} \in  \foreachN{D} \}
  }

  \inferrule[methods-iface]{
    \TYPE\ t_I \ \INTERFACE\ \{ \foreachN{S} \} \in \foreachN{D}
  }
  { \methodSpecifications{\foreachN{D}}{t_I} = \{ \foreachN{S} \}
  }

  \inferrule[sub-struct-refl]{
   }
  { \fgSub{\foreachN{D}}{t_S}{t_S}
  }

  \inferrule[sub-iface]{
   \methodSpecifications{\foreachN{D}}{t} \supseteq \methodSpecifications{\foreachN{D}}{u_I}
  }
  { \fgSub{\foreachN{D}}{t}{u_I}
  }

  \end{mathpar}
\sTransSection{
  \ruleform{\foreachN{D} \turnsFG \reduce{e}{e}}
}{
  Reductions
}
  \bda{llrl}
  \sSynCatName{Value} &
  v & ::= & t_S \{ \foreachN{v} \}
  \\
  \sSynCatName{Evaluation context} &
  \EvCtx & ::= & \Hole
   \mid t_S \{ \foreachN{v}, \EvCtx, \foreachN{e} \}
   \mid \EvCtx.f
   \mid \EvCtx.(t)
   \mid \EvCtx.m(\foreachN{e})
   \mid v.m(\foreachN{v}, \EvCtx, \foreachN{e})
   \\
      \sSynCatName{Substitution (FG values)} &
   \vbFG & ::= & \Angle{\foreachN{\subst{x_i}{v_i}}}
\eda
\begin{mathpar}
 \inferrule[fg-context]
 {\foreachN{D} \turnsFG \reduce{e}{e'}
 }
 {\foreachN{D} \turnsFG \reduce{\EvCtx [e]}{\EvCtx [e']}
 }

 \inferrule[fg-field]
 { \TYPE\ t_S \ \STRUCT\ \{ \foreachN{f \ t} \} \in  \foreachN{D}
   %% \fieldsOf{\foreachN{D}}{t_S} = \foreachN{f \ t}
 }
 { \foreachN{D} \turnsFG \reduce{t_S \{ \foreachN{v} \}.f_i}{v_i}
 }

 \inferrule[fg-call]
  { v = t_S \{ \foreachN{v} \}
    \\ \FUNC\ (x \ t_S) \ m (\foreachN{x \ t}) \ t \ \{ \RETURN\ e \} \in \foreachN{D}
   %% \methodLookup{\foreachN{D}}{(m,\typeOf{v})} = \FUNC\ (x \ t_S) \ m (\foreachN{x \ t}) \ t \ \{ \RETURN\ e \}
 }
 {\foreachN{D} \turnsFG \reduce{v.m(\foreachN{v})}{\Angle{\subst{x}{v}, \foreachN{\subst{x_i}{v_i}}} e}
 }

 \inferrule[fg-assert]
           { v = t_S \{ \foreachN{v} \}
              \\ \fgSub{\foreachN{D}}{t_S}{t}
 }
 { \foreachN{D} \turnsFG \reduce{v.(t)}{v}
 }
\end{mathpar}

}
%% boxfig

A FG program $P$ consists of declarations $\foreachN{D}$ and a main function.
A declaration is either a type declaration for a struct or an interface, or
a method declaration.
Such a method declaration $\FUNC\ (x\ t_S)\ mM\,\{ \RETURN\ e\}$ makes
method of name $m$ and signature $M$ available for receiver type $t_S$,
where the body $e$ may refer to the receiver as $x$.
Expressions $e$ consist of variables $x$, method calls
$e.m(\foreachN{e})$, struct literals $t_S\{\foreachN{e}\}$ with
field values $\foreachN{e}$, access to a struct's
field $e.f$, and dynamic type assertions $e.(t)$.
For convenience, we use disjoint sets of identifiers for structs
$t_S$ and interfaces $t_I$.

% With the exception of variable bindings in function bodies, the primitive type \texttt{int} with operations
% \texttt{==} and \texttt{<}, and the primitive type \texttt{bool} with operations
% \texttt{\&\&} and \texttt{||},
% we can represent the example from Figure~\ref{f:fg} in FG.

FG is a statically typed language.
For brevity, we omit a detailed description of the FG typing rules
as they appear in~\cite{FeatherweightGo} and
will show up in slightly different form in the type-directed translation
in \Cref{sec:type-directed}.
However, we state the following conditions that must be satisfied by a FG program:
\begin{description}
\item[\FGStructNonRec:] Structs must be non-recursive.
\item[\FGUniqueFields:] For each struct, field names must be distinct.
\item[\FGUniqueMethSpec:] For each interface, method names must be distinct.
\item[\FGUniqueReceiver:] Each method declaration is uniquely identified by the receiver type and method name.
\end{description}

The execution of dynamic type assertions in $FG$ relies on structural subtyping.
The relation $\fgSub{\foreachN{D}}{t}{u}$ denotes that under declarations $\foreachN{D}$
type $t$ is a subtype of type $u$ (see~\Cref{f:fg}).
A struct $t_S$ is a subtype of an interface $t_I$ if $t_S$ implements all the methods
specified by $t_I$. An interface $t_I$ is a subtype of another interface $u_I$ if
the methods specified by $t_I$ are a superset of the methods specified by $u_I$.

The dynamic semantics of FG is given in
the bottom part of Figure~\ref{f:fg} as structural operational semantics rules.
The relation $\foreachN{D} \turnsFG \reduce{e}{e'}$ denotes that
expression $e$ reduces to expression $e'$ under the  sequence $\foreachN{D}$ of declarations.
%% MS: fyi
%% Compared to the original presentation~\cite{FeatherweightGo} we explicitly include the in
%% the reduction judgment $\foreachN{D} \turnsFG \reduce{d}{e}$.
Rule \Rule{fg-context} makes use of evaluation contexts $\EvCtx$
with holes $\Hole$ to apply a reduction inside an expression.
Values, ranged over by $v$, are struct literals whose components are all values.
A capture-avoiding substitution $\vbFG = \Angle{\foreachN{\subst{x_i}{v_i}}}$ replaces
variables $x_i$ with values $v_i$, applying a substitution $\vbFG$ to an expression $e$
is written $\vbFG e$.

Rule \Rule{fg-field} deals with field access.
Condition \FGUniqueFields\ guarantees that field name lookup is unambiguous.
Rule \Rule{fg-call} reduces method calls.
Condition \FGUniqueReceiver\ guarantees that method lookup is unambiguous.
The method call is reduced to the method body $e$ where we map the receiver argument to a concrete value $v$
and method arguments $x_i$ to concrete values $v_i$.
Rule \Rule{fg-assert} covers type assertions.
We need to check that the type $t_S$ of value $v$ is consistent with the type $t$ asserted in the program text.
This check can be carried out by checking that $t_S$ and $t$ are in a structural subtype relation.

We write $\reduceFGk{k}{\foreachN{D}}{e}{e'}$ to denote that
$e$ reduces to $e'$
in exactly $k$ steps.
We write $\foreachN{D} \turnsFG \reduceN{e}{e'}$ to denote that there
exists some $k \in \Nat$
with $\reducek{k}{e}{e'}$.
We assume that $\Nat$ includes zero.
If $e$ reduces
ad infinitum, we say that $e$ diverges, written
$\divergeDFG{e}$. Formally,
$\divergeDFG{e}$ if
$\forall k \in \Nat . \exists e' . \foreachN{D} \turnsFG \reducek{k}{e}{e'}$.

FG enjoys type soundness~\cite{FeatherweightGo}. A well-typed program either reduces
to a value, diverges, or it panics by getting stuck on a failed type assertion. The predicate
$\panicDFG{e}$ formalizes panicking:
\begin{mathpar}
  \inferrule*[right=fg-panic]{
    \neg \ \fgSub{\foreachN{D}}{t_S}{t}
  }{
    \panicDFG{\EvCtx [t_S \{ \foreachN{v} \}.(t)]}
  }
\end{mathpar}

%--------------------------------------------------------
%--------------------------------------------------------
\section{Type Directed Translation}
\label{sec:type-directed}

We specify a type-directed translation from FG to
an untyped lambda-calculus extended with recursive let-bindings, constructors, and pattern matching.
The translation itself has already been specified elsewhere~\cite{DBLP:conf/aplas/SulzmannW21},
but there only a weak form of semantic equivalence between source and target
programs was given. The goal of the article at hand is to prove a much stronger form of semantic
equivalence (see \Cref{sec:semantic-preservation}).

%--------------------------------------------------------
\subsection{Target Language}
\label{sec:target-language}

\boxfig{f:target-lang}{Target Language (TL)}{
  \vspace{-2ex}
  \bda{lcr}
  {
    \ba{lll}
    \sSynCatName{Expression} & \expT ::=
    \\ \quad \sSynCatName{Variable}    & \quad \xT \mid \yT         & \mid
    \\ \quad \sSynCatName{Constructor} & \quad \kT                  & \mid
    \\ \quad \sSynCatName{Application} & \quad \expT \ \expT        & \mid \
    \\ \quad \sSynCatName{Abstraction} & \quad \lambda \xT. \expT   & \mid
    \\ \quad \sSynCatName{Pattern case} & \quad \tcaseof{\expT}{[\foreachN{\clsT}]}
    \ea
  }
  &~\quad~&
  {
    \ba{lrl}
     \sSynCatName{Pattern clause} &
     \clsT & ::=  \patT \rightarrow \expT
     \\
     \sSynCatName{Pattern} &
     \patT & ::= \kT \ \foreachN{\xT}
     \\
    \sSynCatName{Program}  &
    \program & ::= {\begin{array}[t]{l}\arraycolsep=0pt
      \LET\ \foreachN{\yT_i = \lambda \xT_i . \expT_i}\\
      \IN\ \expT
      \end{array}}
     \ea
   }
   \eda
  \bda{llrl}
  \sSynCatName{TL values} &
  \uT & ::= & \kT\ \foreachN{\uT} \mid \lambda X . E \mid X
  \\
  \sSynCatName{TL evaluation context} &
  \REvCtxT & ::= & \Hole
   %\mid \kT\  \foreachN{\uT} \REvCtxT \foreachN{\expT}
   \mid \tcaseof{\REvCtxT}{[\foreachN{\patT \rightarrow \expT}]}
   \mid \REvCtxT\ \expT
   \mid \uT\ \REvCtxT
  \\
     \sSynCatName{Substitution (TL values)} &
   \vbTL & ::= & \Angle{\foreachN{\subst{\xT}{\uT}}}
  \\
     \sSynCatName{Substitution (TL methods)} &
   \vbMethTL & ::= & \Angle{\foreachN{\subst{\yT}{\lambda \xT. \expT}}}
   \eda

\vspace{1ex}
\sTransSection{
  \ruleform{\reduceTL{\vbMethTL}{\expT}{\expTA}}
}{
  TL expression reductions
}
\vspace{-0.5ex}
\begin{mathpar}
 \inferrule[tl-context]
 {\reduceTL{\vbMethTL}{\expT}{\expTA}
 }
 {\reduceTL{\vbMethTL}{\REvCtxT [\expT]}{\REvCtxT [\expTA]}
 }

 \inferrule[tl-lambda]
 {}
 { \reduceTL{\vbMethTL}{(\lambda \xT. \expT) \ \uT}{\Angle{\subst{\xT}{\uT}} \expT}
 }

 \inferrule[tl-case]
 { \kT \ \foreach{\xT_i}{n} \rightarrow \expTA \in \foreachN{\clsT}
 }
 {\reduceTL{\vbMethTL}{\tcaseof{\kT \ \foreach{\uT_i}{n}}{{[\foreachN{\clsT}]}}}
         {\Angle{\foreach{\subst{\xT_i}{\uT_i}}{n}} \expTA}
 }

 \inferrule[tl-method]
 {}
 { \reduceTL{\vbMethTL}{\yT \ V}{\vbMethTL(\yT) \ V}
 }

\end{mathpar}

\vspace{1ex}
\sTransSection{
  \ruleform{\reduceTL{}{\program}{\program'}}
}{
  TL reductions
}
\vspace{-0.5ex}
\begin{mathpar}
    \inferrule[tl-prog]
 {\reduceTL{\Angle{\foreachN{\subst{\yT_i}{\lambda \xT_i. \expT_i}}}}{\expT}{\expTA}  }
 { \reduceTL{}{\tletrecin{\foreachN{\yT_i = \lambda \xT_i.\expT_i}}{\expT}}{\tletrecin{\foreachN{\yT_i = \lambda \xT_i. \expT_i}}{\expTA}}
 }

\end{mathpar}
} %% boxfig

Figure~\ref{f:target-lang} specifies the syntax and dynamic semantics of our target language (\TL).
We use capital letters for constructs of the target language. Target expressions $E$ include
variables $X,Y$, data constructors $K$,
function application, lambda abstraction and case expressions to pattern match against constructors.
In a case expression with only one pattern clause,
we often omit the brackets. If a case expressions has more than one clause
$[\foreachN{\patT \rightarrow \expT}]$, we assume that the constructors in $\patT$ are distinct.
A program consists of a sequence of (mutually recursive) function definitions and a (main) expression.
The function definitions are the result of translating FG method definitions.

We assume data constructors for tuples up to some fixed but arbitrary size. The syntax
$(\foreach{E}{n})$ constructs an $n$-tuple when used as an expression, and deconstructs
it when used in a pattern context.
At some places, we use nested patterns as an abbreviation for nested case expressions.
The notation $\lambda \patT. \expT$ stands for
$\lambda X . \tcaseof{X}{\patT \rightarrow \expT}$, where $X$ is fresh.

Target values $V$ consist of constructors, lambda expressions, and variables.
A variable may be a value if it is bound in a $\LET$ at the top-level; that is,
it refers to a method from FG.
A constructor value $K\,\foreach{V}{n}$ is short for $(\ldots (K\,V_1) \ldots)\,V_n$.

The structural operational semantics employs two types of substitutions.
Value substitutions $\vbTL$ records the bindings resulting from pattern matching and function applications.
Method substitutions $\vbMethTL$ records the bindings for translated method definitions (i.e. for
top-level let-bindings).
Reduction of programs is mapped to reduction of expressions under a method substitution,
see rule \Rule{tl-prog}. A variable $Y$ applied to a value $V$ reduces to $\vbMethTL(Y)\,V$ via
\Rule{tl-method}.
The remaining reduction rules are standard.

We write $\reduceTLk{k}{\vbMethTL}{\expT}{E'}$ to denote that
$\expT$ reduces to $E'$
with exactly $k$ steps, and
$\reduceTLN{\vbMethTL}{\expT}{E'}$ for
some finite number of steps.
If $E$ reduces an arbitrary number of steps, we say that
$E$ \emph{diverges}, written $\divergeMTL{E}$.
Formally, $\divergeMTL{E}$ iff
$\forall k \in \Nat . \exists E' . \reduceTLk{k}{\vbMethTL}{E}{E'}$.

In the source language FG, evaluation might panic by getting stuck on a failed type
assertion. The translation to the target language preserves panicking, so we
need to formalize panic. A target language
expression panics if it is stuck on a $\CASE$-expression and there is no
matching clause.
\begin{mathpar}
  \inferrule*[right=tl-panic]{
    \kT \ \foreach{\xT_i}{n} \rightarrow \expTA \not\in [\foreachN{\clsT}]
  }{
    \panicMTL{\REvCtxT[\tcaseof{\kT \ \foreach{\uT_i}{n}}{{[\foreachN{\clsT}]}}]}
  }
\end{mathpar}

%--------------------------------------------------------
\subsection{Translation}

The specification of the translation spreads out over three figures. \Cref{f:type-directed-exp}
gives the translation of expressions, relying on \Cref{f:upcast-downcast} to define auxiliary relations
for translating structural subtyping and type assertions. Finally, \Cref{f:type-directed-meth-prog} translates
method declarations and programs.

Before explaining the translation rules, we establish the following conventions (see also
the top of \Cref{f:type-directed-exp}).
We assume that each FG variable $x$ translates to the TL variable $\xT$.
FG variables introduced in method declarations are assumed to be distinct.
This guarantees that there are no name clashes in environment $\fgEnv$.
For each struct $t_S$ we introduce a TL constructor $K_{t_S}$, and
for each interface $t_I$ we introduce a TL constructor $K_{t_I}$.
For each method declaration $\FUNC\ (x \ {t_S}) \ m M \ \{ \RETURN\ e \}$ we introduce
a TL variable $\mT{m}{t_S}$, thereby relying on condition \FGUniqueReceiver{} which guarantees that $m$ and
$t_S$ uniquely identify this declaration.
We write $\fgEnv$ to denote typing environments where we record the types of FG variables.
The notation $[n]$ is a short-hand for the set $\{1,\dots,n\}$.

The overall idea of the translation is to choose the TL-representation of an FG-value $v = t_S\{\foreachN{v}\}$
based on the type $t$ that $v$ is used at:

\begin{itemize}
\item If $t$ is a struct type $t_S$, then the representation of $v$ is $K_{t_S}\,(\foreachN{V})$, where
  each $V_i$ is the representation of $v_i$, so $(\foreachN{V})$ is a tuple for the struct fields.
\item If $t$ is an interface type, then the representation of $v$ is an \emph{interface-value}
  $K_{t_I}\,(V, \foreachN{\mT{m_i}{t_S}})$, where $V$ is the representation of $v$ at struct type $t_S$
  and $\foreachN{\mT{m_i}{t_S}}$ is a dictionary~\cite{Hall:1996:TCH:227699.227700} containing all
  methods $\foreachN{m_i}$ of interface $t_I$. The translation makes each method
  $\FUNC\ (x \ {t_S}) \ m_i M \ \{ \RETURN\ e \}$
  available as a top-level binding $\LET\ \mT{m_i}{t_S} = E$.
  An interface value $K_{t_I}\,(V, \foreachN{\mT{m_i}{t_S}})$ bears close resemblance to an
  existential type~\cite{Laufer:1994:PTI:186025.186031}, as it hides the concrete representation of the value $V$.
\end{itemize}

\boxfig{f:type-directed-exp}{Translation of expressions}{
  Convention for mapping source to target terms
  \vspace{-1ex}
  \begin{mathpar}
    x \leadsto \xT
    \quad
    t_S \leadsto K_{t_S}
    \quad
    t_I \leadsto K_{t_I}
    \quad
    \FUNC\ (x \ {t_S}) \ m M \ \{ \RETURN\ e \} \leadsto \mT{m}{t_S}
  \end{mathpar}
  \bda{llcl}
\mbox{FG environment} & \fgEnv & ::= & \{ \} \mid \{ x : t \} \mid \fgEnv \cup \fgEnv
  \eda
  \vspace{0.1ex}

  \sTransSection{
    \ruleform{\tdExpTransGrayBox{\pair{\foreachN{D}}{\fgEnv}}{e : t}{\expT}}
  }{Translating expressions}
  \begin{mathpar}
  \inferrule[td-var]
            { (x : t) \in \fgEnv
            }
            { \tdExpTransGrayBox{\pair{\foreachN{D}}{\fgEnv}}{x : t}{\xT}
            }

  \inferrule[td-struct]
  {\TYPE\ t_S \ \STRUCT\ \{ \foreach{f \ t}{n} \} \in  \foreachN{D}
               \\\\   \tdExpTransGrayBox{\pair{\foreachN{D}}{\fgEnv}}{e_i : t_i}{\expT_i}\quad\noteForall{i \in [n]}
            }
            { \tdExpTransGrayBox{\pair{\foreachN{D}}{\fgEnv}}{t_S \{ \foreach{e}{n} \} : t_S }{\kT_{t_S} \ (\foreach{\expT}{n})}
            }

  \inferrule[td-access]
  {\tdExpTransGrayBox{\pair{\foreachN{D}}{\fgEnv}}{e : t_S}{\expT}
             \\  \TYPE\ t_S \ \STRUCT\ \{ \foreach{f \ t}{n} \} \in  \foreachN{D}
          }
           { \tdExpTransGrayBox{\pair{\foreachN{D}}{\fgEnv}}
                        {e.f_i : t_i }
                        { \CASE\ \expT\ \OF\ \kT_{t_S} \ (\foreach{\xT}{n}) \rightarrow \xT_i}
           }

  \inferrule[td-call-struct]
            { m (\foreach{x \ t}{n}) \ t \in \methodSpecifications{\foreachN{D}}{t_S}
              \\
              \tdExpTransGrayBox{\pair{\foreachN{D}}{\fgEnv}}{e : t_S}{\expT}
              \\ \tdExpTransGrayBox{\pair{\foreachN{D}}{\fgEnv}}{e_i : t_i}{\expT_i}\quad\noteForall{i \in [n]}
          }
          { \tdExpTransGrayBox{\pair{\foreachN{D}}{\fgEnv}}{e.m(\foreach{e}{n}) : t}{\mT{m}{t_S} \ \expT \ (\foreach{\expT}{n}) } }

  \inferrule[td-call-iface]
  {\tdExpTransGrayBox{\pair{\foreachN{D}}{\fgEnv}}{e : t_I}{\expT}
    \\ \TYPE\ t_I \ \INTERFACE\ \{ \foreachN{S} \} \in \foreachN{D}
    \\ S_j = m(\foreach{x \ t}{n})\,t
    \\ \tdExpTransGrayBox{\pair{\foreachN{D}}{\fgEnv}}{e_i : t_i}{\expT_i}\quad\noteForall{i \in [n]}
    \\ X,\foreach{X}{q}\textrm{~fresh}
  }
  { \tdExpTransGrayBox{\pair{\foreachN{D}}{\fgEnv}}
    {e.m(\foreach{e}{n}) : t}
    {\CASE\ \expT \ \OF\ \kT_{t_I} \ (\xT, \foreach{\xT}{q}) \rightarrow \xT_j \ \xT \ (\foreach{\expT}{n}) }
  }

  \inferrule[td-assert]
  {
    u ~\textrm{defined in}~ \foreachN{D}
    \\\\
    \tdExpTransGrayBox{\pair{\foreachN{D}}{\fgEnv}}{e : t_I}{\expT_2}
          \\ \tdDowncastGrayBox{\foreachN{D}}{\assertOf{t_I}{u}}{\expT_1}
          }
          { \tdExpTransGrayBox{\pair{\foreachN{D}}{\fgEnv}}{e.(u) : u}{\expT_1 \ \expT_2} }

  \inferrule[td-sub]
            {\tdExpTransGrayBox{\pair{\foreachN{D}}{\fgEnv}}{e : t}{\expT_2}
          \\\\ \tdUpcastGrayBox{\foreachN{D}}{\subtypeOf{t}{u}}{\expT_1}
          }
            { \tdExpTransGrayBox{\pair{\foreachN{D}}{\fgEnv}}{e : u}{\expT_1 \ \expT_2} }
  \end{mathpar}

}
%% fig

The translation rules for expressions (\Cref{f:type-directed-exp}) are of the form
$\tdExpTrans{\pair{\foreachN{D}}{\fgEnv}}{e : t}{\expT}$
where $\foreachN{D}$ refers to the sequence of FG declarations,
$\fgEnv$ refers to type binding of local variables,
$e$ is the to be translated FG expression,
$t$ its type and $\expT$ the resulting target term.
Rule \Rule{td-var} translates variables and follows our convention that $x$ translates to $\xT$.
Rule \Rule{td-struct} translates a struct creation. The translated field elements $\expT_i$
are collected in a tuple and tagged via the constructor $K_{t_S}$.
Rule \Rule{td-access} uses pattern matching to capture field access in the translation.

Method calls are dealt with by rules \Rule{td-call-struct} and \Rule{td-call-iface}.
Rule \Rule{td-call-struct} covers the case that the receiver $e$ is of the struct type $t_S$.
The first precondition guarantees that an implementation for this specific method call exists.
(See Figure~\ref{f:fg} for the auxiliary $\methodSpecificationsRel$.)
Hence, we can assume that we have available a corresponding definition for $\mT{m}{t_S}$ in our translation.
The method call then translates to applying $\mT{m}{t_S}$ first on the translated receiver $\expT$,
followed by the translated arguments collected in a tuple $(\foreach{\expT}{n})$.

Rule \Rule{td-call-iface} assumes that receiver $e$ is of interface type $t_I$,
so $e$ translates to interface-value $E$.
Hence, we pattern match on $E$
to access the underlying value and the desired method in the dictionary.
We assume that the order of methods in the dictionary corresponds to the order
of method declarations in the interface.
The preconditions guarantee that $t_I$ provides a method $m$ as demanded by the method call,
where $j$ denotes the index of $m$ in interface $t_I$.

To explain the two remaining rules for expressions (\Rule{td-assert} and \Rule{td-sub}),
we first introduce the two auxiliary relations defined in \Cref{f:upcast-downcast}.
The relation $\tdUpcast{\foreachN{D}}{\subtypeOf{t}{u_I}}{\expT}$ constructs an interface-value
for $u_I$. Thus, the resulting expression $\expT$ is a $\lambda$-expression taking the representation of a value at type $t$
and yields its representation at type $u_I$.

\boxfig{f:upcast-downcast}{Translation of structural subtyping and type assertions}{

  \sTransSection{
  \ruleform{\tdUpcastGrayBox{\foreachN{D}}{\subtypeOf{t}{u_I}}{\expT}}
}{
  Interface-value construction
}
\begin{mathpar}
\inferrule[td-cons-struct-iface]
          {\TYPE\ t_I \ \INTERFACE\ \{ \foreachN{S} \} \in \foreachN{D}
        \\ \methodSpecifications{\foreachN{D}}{t_S} \supseteq \foreachN{S}
        \\ \foreachN{S} = \foreach{m M}{n}
          }
          {\tdUpcastGrayBox{\foreachN{D}}{\subtypeOf{t_S}{t_I}}
                    {\lambda \xT.
                      \kT_{t_I} \ (\xT, \foreach{\mT{m_i}{t_S}}{n}) }
         }

\inferrule[td-cons-iface-iface]
          {\TYPE\ t_I \ \INTERFACE\ \{ \foreach{R}{n} \} \in \foreachN{D}
       \\ \TYPE\ u_I \ \INTERFACE\ \{ \foreach{S}{q} \} \in \foreachN{D}
       \\  S_i = R_{\mapPerm(i)} \quad\noteForall{i \in [q]}
        }
          {\tdUpcastGrayBox{\foreachN{D}}{\subtypeOf{t_I}{u_I}}
                    { \lambda \xT. \CASE\,\xT\,\OF\
                          \kT_{t_I} \ (\xT, \foreach{\xT}{n})
                          \rightarrow
                          \kT_{u_I} \ (\xT, \xT_{\mapPerm(1)}, \ldots, \xT_{\mapPerm(q)})
                    }
         }
\end{mathpar}

\sTransSection{
  \ruleform{\tdDowncastGrayBox{\foreachN{D}}{\assertOf{t_I}{u}}{E}}
}{
  Interface-value destruction
}
\begin{mathpar}
  \inferrule[td-destr-iface-struct]
            {\TYPE\ t_I \ \INTERFACE\ \{ \foreach{R}{n} \} \in \foreachN{D}
             \\ \fgSub{\foreachN{D}}{t_S}{t_I}
            }
            {\tdDowncastGrayBox{\foreachN{D}}{\assertOf{t_I}{t_S}}{\lambda \xT.  \CASE\,\xT\,\OF\
                    \kT_{t_I} \ (K_{t_S} \ Y, \foreach{\xT}{n})
                      \rightarrow K_{t_S} \, Y
                   }
            }

\inferrule[td-destr-iface-iface]
{ X, Y, Y', \foreach{X}{n}~\textrm{fresh}\\
  \TYPE\ t_I \ \INTERFACE\ \{ \foreach{R}{n} \} \in \foreachN{D}
            \\
            \GrayBox{
              \textrm{for all}~\TYPE\ t_{Sj}\ \STRUCT\ \{ \foreachN{f \ u} \} \in \foreachN{D}
              ~\textrm{with}~ \tdUpcast{\foreachN{D}}{\subtypeOf{t_{Sj}}{u_I}}{\expT_j}
              \textrm{:}
            }
            \\
            \GrayBox{
              \clsT_j =
              \kT_{t_{Sj}} \ Y' \rightarrow (\expT_j \ (\kT_{t_{Sj}} \ Y'))
            }
          }
          { \tdDowncastGrayBox{\foreachN{D}}{\assertOf{t_I}{u_I}}{\lambda \xT.
                             \CASE\ \xT \ \OF\
                             \kT_{t_I}\, (\yT, \foreach{\xT}{n}) \rightarrow \CASE\ \yT  \ \OF\
                             [\foreachN{\clsT}]}
          }
\end{mathpar}
}

The preconditions in rule \Rule{td-cons-struct-iface} check that struct $t_S$ implements the interface.
This guarantees the existence of method definitions $\mT{m_i}{t_S}$.
Hence, we can construct the desired interface-value.
The preconditions in rule \Rule{td-cons-iface-iface} check that $t_I$'s methods are a superset of $u_I$'s methods.
This is done via the total function $\mapPerm : \{1,\ldots,q\} \to \{1,\ldots,n\}$ that matches each (wanted) method in $u_I$ against a (given) method in $t_I$.
We use pattern matching over the $t_I$'s interface-value to extract the wanted methods.
Recall that dictionaries maintain the order of method as specified by the interface.

The relation $\tdDowncast{\foreachN{D}}{\assertOf{t_I}{u}}{E}$ destructs an interface-value.
The $\lambda$-expression $E$ takes a representation at type $t_I$ and converts it to the
representation at type $u$. This conversion might fail, resulting in a pattern-match error.

Rule \Rule{td-destr-iface-struct} deals with the case that the target type $u$ is a struct type $t_S$.
Hence, we find the precondition $\fgSub{\foreachN{D}}{t_S}{t_I}$.
We pattern match over the interface-value that represents $t_I$ to check that the underlying value matches $t_S$
and extract the value.
It is possible that some other value has been used to implement the interface-value that represents $t_I$.
In such a case, the pattern match fails and we experience run-time failure.

Rule \Rule{td-destr-iface-iface} deals with the case that the target type $u$ is an interface type $u_I$.
The outer case expression extracts the value $Y$ underlying interface-value $t_I$
(this case never fails).
We then check if we can construct an interface-value for $u_I$ via $Y$.
This is done via an inner case expression.
For each struct $t_{Sj}$ implementing $u_I$,
we have a pattern clause $\clsT_j$ that
matches against the constructor $\kT_{t_{Sj}}$ of the struct and then constructs an interface-value for $u_I$.
There are two reasons for run-time failure here.
First, $\yT$ (used to implement $t_I$) might not implement $u_I$;
that is, none of the pattern clauses $\clsT_j$ match.
Second, $[\foreachN{\clsT}]$ might be empty because
no receiver at all implements $u_I$. This case is rather unlikely
and could be caught statically.

We now come back explaining the two remaining translation rules for expressions (\Cref{f:type-directed-exp}).
Rule \Rule{td-assert} translates a type assertion $e.(u)$ by destructing $e$'s interface-value, potentially
yielding a representation at type $u$. The type of $e$ must be an interface type because only conversions
from an interface type to some other type must be checked dynamically. Rule \Rule{td-sub} translates
structural subtyping by constructing an appropriate interface-value.
This rule could be integrated as part of the other rules to make the translation more syntax-directed.
For clarity, we prefer to have a stand-alone subtyping rule.

\boxfig{f:type-directed-meth-prog}{Translation of methods and programs}{

\sTransSection{
  \ruleform{\tdMethTransGrayBox{\foreachN{D}}{\FUNC\ (x \ t_S) \ m (\foreachN{x \ t}) \ t}{\expT} }
}{
  Translating method declarations
}
  \bda{c}
  \inferruleLeft[td-method]
         { \tdExpTransGrayBox{\pair{\foreachN{D}}{\{ x : t_S, \foreach{x_i : t_i}{n} \}}}{e : t}{\expT}
          }
          { \tdMethTransGrayBox{\foreachN{D}}
                        {\FUNC\ (x \ t_S) \ m (\foreach{x \ t}{n}) \ t \ \{ \RETURN\ e \}}
                        {\lambda \xT . \lambda (\foreach{\xT }{n}). \expT}
          }

  \eda

\sTransSection{
  \ruleform{\tdProgTransGrayBox{P}{\program}}
}{
  Translating programs
}
  \bda{c}
 \inferrule[td-prog]
 { \textrm{all types used in~}\foreachN{D}~\textrm{are defined in}~\foreachN{D}\\
   \tdExpTransGrayBox{\pair{\foreachN{D}}{\EmptyFgEnv}}{e : t}{\expT} \\\\
   \tdMethTransGrayBox{\foreachN{D}}{D_i'}{\expT_i}\\
   \\D_i' = \FUNC\ (x_i \ {t_S}_i) \ m_i M_i \ \{ \RETURN\ e_i \}\\
   (\textrm{for all}~i \in [n],
   \textrm{where}~\foreach{D'}{n}~\textrm{are the}~\FUNC~\textrm{declarations in}~\foreachN{D})\\
 }
 { \tdProgTransGrayBox{\ignore{\PACKAGE\ \MAIN;} \foreachN{D} \ \FUNC\ \MAIN () \{ \_ = e \}}
   { \LET\ \foreach{\mT{m_i}{{t_S}_i} = \expT_i}{n} \ \IN\ \expT}
 }
  \eda
}

The translation of programs and methods (\Cref{f:type-directed-meth-prog})
boils down to the translation of expressions involved.
Rule \Rule{td-method} translates a specific method declaration,
rule \Rule{td-prog} collects all method declarations and also translates the main expression.
The type system induced by the translation rules is equivalent to the original type system
of Featherweight Go. See the Appendix.

%--------------------------------------------------------
%--------------------------------------------------------
\section{Semantic preservation}
\label{sec:semantic-preservation}
\label{sec:logical-relations}
\label{sec:properties}

\boxfig{f:reduce-rel-fg-tl}{Relating FG to \TL\ Reduction}{

  \sTransSection{
   \ruleform{\redLRConfk{k}{t}{e}{\expT}}
  }{FG versus \TL\ expressions}
\begin{mathpar}
 \inferrule[red-rel-exp]{
   {
     \ba{c}
     \left(
       \forall k' < k,
       v ~.~ \reduceFGk{\kA}{\foreachN{D}}{e}{v} \implies \exists \uT.
       \reduceTLN{\vbMethTL}{\expT}{\uT} \wedge \redVLRConfk{k- \kA}{t}{v}{\uT}
     \right)
     \\
     \wedge
     \\
     \left(
       \forall k' < k,
       e' ~.~ \reduceFGk{k'}{\foreachN{D}}{e}{e'} \wedge \divergeDFG{e'}
       \implies
       \divergeMTL{\expT}
       \right)
     \PANIC{
     \\
     \wedge
     \\
     \left(
       \forall k' < k,
       e' ~.~ \reduceFGk{\kA}{\foreachN{D}}{e}{e'} \wedge \panicFG{\foreachN{D}}{e'} \implies
              \panicMTL{\expT}
       \right)
     }
     %% PANIC
     \ea
   }
 }{
   \redLRConfk{k}{t}{e}{\expT}
 }
\end{mathpar}

  \sTransSection{
   \ruleform{\redVLRConfk{k}{t}{v}{\uT}}
  }{FG versus \TL\ values}
\begin{mathpar}
 \inferrule[red-rel-struct]
     { \TYPE\ t_S \ \STRUCT\ \{ \foreach{f \ t}{n} \} \in  \foreachN{D}
       \\ \forall i \in [n] . \redVLRConfk{k}{t_i}{v_i}{\uT_i}
     }
     { \redVLRConfk{k}
       {t_S}
       {t_S \{ \foreach{v}{n} \}}
       {\kT_{t_S} \ (\foreach{\uT}{n})}
     }

 \inferrule[red-rel-iface]
    {\uT = \kT_{u_S} \ \foreachN{\uTA}
      \\ \forall k_1 \leq k . \redVLRConfk{k_1}{u_S}{v}{\uT}
      \\ \methodSpecifications{\foreachN{D}}{t_I} = \{ \foreach{m M}{n} \}
      \\ \forall k_2 \leq k, i \in [n] . \redLRConfk{k_2}
                       {m_i M_i}
                       {\methodLookup{\foreachN{D}}{(m_i,u_S)}}
                       {\yT_i}
     }
     { \redVLRConfk{k}{t_I}{v}
             {\kT_{t_I} \ (\uT, \foreach{\yT}{n})}
     }

 \end{mathpar}

  \sTransSection{
    \ruleform{\redLRConfk{k}
                    {m M}
                    {\FUNC\ (x \ t_S) \ mM \ \{ \RETURN\ e \}}{\yT}}
  }{FG versus \TL\ methods}
 \begin{mathpar}
 \inferrule[red-rel-method]
           { \forall \kA < k,
                v', \uTA, \foreach{v_i}{n}, \foreach{\uT_i}{n}.
                    (\redVLRConfk{\kA}{t_S}{v'}{\uTA}
                    \wedge (\forall i \in [n]. \redVLRConfk{\kA}{t_i}{v_i}{\uT_i}))
         \\ \implies \redLRConfk{\kA}
               {t}
               {\Angle{\subst{x}{v'},\foreach{\subst{x_i}{v_i}}{n}} e}
               {(\yT \ \uTA) \ (\foreach{\uT}{n})}
 }
     { \redLRConfk{k}
       {m (\foreach{x \ t}{n}) \ t}
       {\FUNC\ (x \ t_S) \ m (\foreach{x \ t}{n}) \ t \ \{ \RETURN\ e \}}
       {\yT}
 }
\end{mathpar}

  \sTransSection{
    \ruleform{\redLRk{k}{\triple{\foreachN{D}}{\vbMethTL}{\fgEnv}}
                    {}
                    {\vbFG}
                    {\vbTL}}
  }{FG environments versus \TL\ value substitutions}
  \begin{mathpar}
    \inferrule[red-rel-vb]
              { \forall (x : t) \in \fgEnv .
                 \redLRConfk{k}{t}{\vbFG(x)}{\vbTL(X)}
              }
              {
                \redLRk{k}{\triple{\foreachN{D}}{\vbMethTL}{\fgEnv}}
                    {}
                    {\vbFG}
                    {\vbTL}
             }
  \end{mathpar}

  \sTransSection{
    \ruleform{\redLRk{k}
                    {}
                    {}
                    {\foreachN{D}}
                    {\vbMethTL}}
  }{FG declarations versus \TL\ method substitutions}
 \begin{mathpar}
   \inferrule[red-rel-decls]
             { \forall\,
               \FUNC\ (x \ t_S) \ m M \ \{ \RETURN\ e \} \in \foreachN{D} :\\
              \redLRConfk{k}
                    {m M}
                    {\FUNC\ (x \ t_S) \ mM \ \{ \RETURN\ e \}}{\mT{m}{t_S}}
             }
             { \redLRk{k}
                    {}
                    {}
                    {\foreachN{D}}
                    {\vbMethTL}
             }

 \end{mathpar}
}
%% fig

We establish correctness of the type-directed translation scheme
by showing that source and target behave the same.
Figure~\ref{f:reduce-rel-fg-tl} introduces the details of the logical relations
that are discussed in the earlier Section~\ref{sec:overview}.
We assume that step indices $k$ are natural numbers starting with $0$.

The relation $\redLRConfk{k}{t}{e}{\expT}$ specifies how an FG expression $e$ and TL expression
$E$ are related at FG type $t$.
The three cases \Rule{terminate}, \Rule{diverge} and \Rule{panic}
from \Cref{f:highlights} are combined in one rule \Rule{red-rel-exp}.

The relation $\redVLRConfk{k}{t}{v}{\uT}$ specifies when FG value $v$ and TL value $V$
are equivalent at FG type $t$.
Rule \Rule{red-rel-struct} covers struct values by ensuring that the constructor tag
matches and all field values are equivalent.
Rule \Rule{red-rel-iface} generalizes \Rule{iface-$\leq$}
from Figure~\ref{f:step-index}. The auxiliary $\methodLookupSym$ retrieves
the method declaration for some method name and receiver type:
\bda{c}\small
  \myirule{ \FUNC\ (x \ t_S) \ mM \ \{ \RETURN\ e \} \in  \foreachN{D} }
          { \methodLookup{\foreachN{D}}{(m,t_S)} = \FUNC\ (x \ t_S) \ mM \ \{ \RETURN\ e \} }
\eda

Rule \Rule{red-rel-method} relates methods, generalizing \Rule{method-$<$}
from \Cref{sec:overview}.
Rules \Rule{red-rel-vb} and \Rule{red-rel-decls} lift the logical relation
to environments and declarations.

Thanks to the step index our logical relations are well-founded. In the definitions
in \Cref{f:reduce-rel-fg-tl},
the relations $\redLRConfk{k}{t}{e}{\expT}$ and $\redVLRConfk{k}{t}{v}{\uT}$
form a cycle. But either the step index $k$ decreases or it stays constant but the size of the target value $V$
decreases in recursive calls.
Several other basic properties hold, such as
\Rule{lr-step} and \Rule{lr-mono} from the earlier Figure~\ref{f:step-index}.
Details are given in the appendix.
These properties are vital to establish the following results.

We can prove that target expressions resulting from FG expressions
are semantically equivalent to the source.

\begin{lemma}[Expression Equivalence]
\label{le:exp-red-equivalent}
  Let $\tdExpTrans{\pair{\foreachN{D}}{\fgEnv}}{e : t}{\expT}$
  and $\vbFG$, $\vbTL$, $\vbMethTL$
  such that
  $\redLRk{k}{\triple{\foreachN{D}}{\vbMethTL}{\fgEnv}}{}{\vbFG}{\vbTL}$
  and
  $\redLRk{k}{}{}{\foreachN{D}}{\vbMethTL}$ for some $k$.
  Then, we find that $\redLRConfk{k}{t}{\vbFG(e)}{\vbTL(\expT)}$.
  \end{lemma}

As motivated in Section~\ref{sec:overview}, for the proof to go through,
one of the rules \Rule{red-rel-iface} and \Rule{red-rel-method}
must use $<$ and the other $\leq$.
In our case, we use  $\leq$ in rule \Rule{red-rel-iface} and $<$ in rule \Rule{red-rel-method}.
The lengthy proof is given in the appendix.

Based on the above result, we can establish semantic equivalence for method definitions.
For this proof to go through it is essential
that we find $\leq$ in rule \Rule{red-rel-iface} and $<$ in rule \Rule{red-rel-method}.

\begin{lemma}[Method Equivalence]
  \label{le:method-red-rel-equiv}
  \mbox{} \\ %% margin
  Let $\foreachN{D}$ and $\vbMethTL$ such that
  for each $\FUNC\ (x \ t_S) \ m (\foreach{x \ t}{n}) \ t \ \{ \RETURN\ e \}$ in $\foreachN{D}$
  we have $\vbMethTL(\mT{m}{t_S}) = \lambda \xT . \lambda (\foreach{\xT }{n}). \expT$ where
  $\tdMethTrans{\foreachN{D}}
  {\FUNC\ (x \ t_S) \ m (\foreach{x \ t}{n}) \ t \ \{ \RETURN\ e \}}
  {\lambda \xT . \lambda (\foreach{\xT }{n}). \expT}$.
  Then, we find that $\redLRk{k}{}{}{\foreachN{D}}{\vbMethTL}$ for any $k$.
\end{lemma}
\begin{proof}
   To verify (1) $\redLRk{k}{}{}{\foreachN{D}}{\vbMethTL}$ for
   each
   $
   \FUNC\ (x \ t_S) \ m (\foreach{x_i \ t_i}{n}) \ t \ \{ \RETURN\ e \} \ \mbox{in} \ \foreachN{D}
   $
    we have to show based on rules \Rule{red-rel-decls} and \Rule{red-rel-method} that
   \bda{c}
     \forall \kA < k,
                v, \uT, \foreach{v}{n}, \foreach{\uT}{n}.
                    (\redLRConfk{\kA}{t_S}{v}{\uT}
                    \wedge (\forall i \in [n]. \redLRConfk{\kA}{t_i}{v_i}{\uT_i}))
         \\[1.5mm] \implies (2) \ \redLRConfk{\kA}
               {t}
               {\Angle{\subst{x}{v},\foreach{\subst{x_i}{v_i}}{n}} e}
               {(\mT{m}{t_S} \ \uT) \ (\foreach{\uT}{n})}
    \eda

   \noindent
   We verify the result by induction on $k$.
   \begin{itemize}
   \item
       {\bf Case} $k=0$ or $k=1$: Holds immediately. See rule \Rule{red-rel-exp}.

     \item
     {\bf Case} $k \implies k+1$:
     Suppose $\kA < k+1$
     and (3) $\redLRConfk{\kA}{t_S}{v}{\uT}$
     and (4) $\redLRConfk{\kA}{t_i}{v_i}{\uT_i}$
     for some $v$, $\uT$, $v_i$, $\uT_i$ for $i \in [n]$.
     Define $\vbFG = \Angle{\subst{x}{v},\foreach{\subst{x_i}{v_i}}{n}}$
     and $\vbTL = \Angle{\subst{\xT}{\uT},\foreach{\subst{\xT_i}{\uT_i}}{n}}$
     and $\fgEnv = \{ x : t_S, \foreach{x_i : t_i}{n} \}$.
     \\[\medskipamount]
     (5)   $\redLRk{\kA}{\triple{\foreachN{D}}{\vbMethTL}{\fgEnv}}{}{\vbFG}{\vbTL}$
     via (3) and (4).
     \\[1mm]
     (6) $\redLRk{\kA}{}{}{\foreachN{D}}{\vbMethTL}$ by induction.
     \\[1mm]
     (7) $\tdExpTrans{\Angle{\foreachN{D},\fgEnv}}{e : t}{E}$ from the assumption and rule \Rule{td-method}.
     \\
     (8) $\redLRConfk{\kA}{t}{\vbFG e}{\vbTL E}$
     via (5), (6), (7), and Lemma~\ref{le:exp-red-equivalent}.
     \\[0.5mm]
     (9) $\reduceTLN{\vbMethTL}{(\mT{m}{t_S} \ \uT) \ (\foreach{\uT}{n})}{\vbTL \expT}$
     \\\phantom{xx}
     via the assumption that $\vbMethTL(\mT{m}{t_S}) = \lambda \xT . \lambda (\foreach{\xT }{n}). \expT$.
     \\
     (10) $\redLRConfk{\kA}
               {t}
               {\vbFG e}
               {(\mT{m}{t_S} \ \uT) \ (\foreach{\uT}{n})}$
     \\\phantom{xx}
     via (8), (9) and because target reductions do not affect the step index\\\phantom{xx}
     (Lemma~\ref{le:target-step-preserve-equivalence} in the Appendix).
     \\[\smallskipamount]
    %% MS: see Lemma~\ref{le:target-step-preserve-equivalence}
    Statement (10) corresponds to (2).
    Thus, we can establish (1). \qed
  \end{itemize}
\end{proof}

If we would find $\leq$ instead of $<$ in rule \Rule{red-rel-method},
the proof would not go through. We would then need to establish
the implication at the beginning of the proof for $k' \leq k$, but
the induction hypothesis gives us only
$\redLRk{k' - 1}{}{}{\foreachN{D}}{\vbMethTL}$ in (6).

We state our main result that the dictionary-passing translation
preserves the dynamic behavior of FG programs.

\begin{theorem}[Program Equivalence]
\label{theo:prog-equiv}
Let $\tdProgTrans{\ignore{\PACKAGE\ \MAIN;} \foreachN{D} \ \FUNC\ \MAIN () \{ \_ = e \}}
{\LET\ \foreach{\mT{m_i}{{t_S}_i} = \expT_i}{n} \ \IN\ \expT}$
where we assume that $e$ has type $t$.
Let $\vbMethTL = \Angle{\foreach{\subst{\mT{m_i}{{t_S}_i}}{\expT_i}}{n}}$.
Then, we find that
$\redLRConfk{k}{t}{e}{\expT}$
for any $k$.
\end{theorem}
\begin{proof}
  Follows from Lemmas~\ref{le:exp-red-equivalent} and~\ref{le:method-red-rel-equiv}.
  \qed
\end{proof}

%--------------------------------------------------------
%--------------------------------------------------------
\section{Related Work}
\label{sec:related}

Logical relations have a long tradition of proving properties of
typed programming languages.
Such properties include termination~\cite{DBLP:journals/jsyml/Tait67,journals/iandc/Statman85},
type safety~\cite{Skorstengaard}, and program equivalence~\cite[Chapters~6,~7]{ATAPL}.
A logical relation (LR) is often defined inductively, indexed by type.
If its definition is based on an operational semantics, the LR is called
syntactic~\cite{conf/icalp/Pitts98,journals/entcs/CraryH07}.
With recursive types, a step-index~\cite{journals/toplas/AppelM01,10.1007/11693024_6}
provides a decreasing measure
to keep the definition well-founded.
See \cite[Chapter~8]{books/daglib/0085577} and \cite{Skorstengaard} for introductions to the topic.

LRs are often used to relate two terms of the same language. For our translation, the two terms
are from different languages, related at a type from the source language.
Benton and Hur~\cite{conf/icfp/BentonH09} prove correctness of compiler transformations.
They used a step-index LR
to relate a denotational semantics of the $\lambda$-calculus with recursion to configurations of a SECD-machine.
The setup relies on
biorthogonality~\cite{journals/mscs/Pitts00,conf/lics/MelliesV05,conf/dagstuhl/Pitts10,jaber:hal-00594386}
to allow for compositionality and extensionality of equivalences.

Hur and Dreyer~\cite{conf/popl/HurD11} build on this idea to show equivalence between
an expressive source language (polymorphic $\lambda$-calculus with references, existentials, and recursive types)
and assembly language. Their biorthogonal, step-indexed Kripke LR does not directly relate the two languages
but relies on abstract language specifications. The Kripke part of the LR~\cite{Pitts98operationalreasoning}
allows reasoning about the shape of the heap.

Our setting is different in that we consider a source language with support for overloading.
Besides structured data and functions, we need to cover interface values.
This then leads to some challenges to get the step index right.
Recall Figure~\ref{f:step-index} and the discussion in Section~\ref{sec:overview}.

Simulation or bisimulation (see e.g. \cite{journals/jacm/SumiiP07}) is another common technique for showing
program equivalences. In our setting, using this technique amounts to proving that reduction and translation
commutes:
if source term $e$ reduces to $e'$ and translates to target term $E$, then $e'$ translates
to $E'$ such that $E$ reduces to $E''$ (potentially in several steps) with $E' = E''$.
One challenge is that two target
terms $E'$ and $E''$ are not necessarily syntactically equal but only semantically.
With LR, we abstract away certain details of single step reductions, as we only compare values not intermediate
results. A downside of the LR is that getting the step index right is sometimes not trivial.
%% MS omit: felt like saying, of course, we could use here a simulation (but I'm not sure).
%% However, we would like to extend our translation to other settings as well, for example
%% to provide a type-directed translation for Featherweight Go extended with generics~\cite{FeatherweightGo}.
%% In such a setting, the two target
%% terms $E'$ and $E''$ are not necessarily syntactically equal but only semantically.

Paraskevopoulou and Grover~\cite{journals/pacmpl/Paraskevopoulou21a} combine simulation
and an untyped, step-indexed LR~\cite{conf/popl/AcarAB08}
to relate the translation of a reduced expression (the $E'$ from the preceding paragraph)
with the reduction result of the translated expression (the $E''$). They use this technique
to prove correctness of CPS transformations using small-step and big-step operational
semantics. Resource invariants connect the number of steps
a term and its translation might take, allowing them to prove that
divergence and asymptotic runtime is preserved by the transformation. Our LR
does not support resource invariants but includes a case for divergence directly.

Hur and coworkers~\cite{conf/popl/HurDNV12}
as well as
Hermida and coworkers~\cite{hermida_reddy_robinson_santamaria_2022}
also blend bisimulation with LRs, building
on previous results~\cite{conf/popl/HurD11}.

%--------------------------------------------------------
%--------------------------------------------------------
\section{Conclusion}
\label{sec:conc}

In this work, we established a strong semantic preservation result for a type-directed translation scheme of Featherweight Go.
To achieve this result, we rely on syntactic, step-indexed logical relations. There are some subtle corners and
we gave a detailed discussion of how to get the definition of logical relations right
so that the proofs will go through. The proofs are still hand-written where all cases are worked out in detail.
To formalize the proofs in a proof assistant we yet need to mechanize the source and target semantics.
This is something we plan to pursue in future work.

We believe that the methods developed in this work
will be useful in other language settings that employ a type-directed translation scheme
for a form of overloading, e.g.~consider Haskell type classes~\cite{Hall:1996:TCH:227699.227700} and
traits in Scala~\cite{scala} and Rust~\cite{rust}.
This is another topic for future work.
In another direction, we plan to adapt our translation scheme
and proof method to cover Featherweight Go extended with generics~\cite{FeatherweightGo}.

\paragraph{Acknowledgments.}
We thank the MPC'22 reviewers for their comments.

\bibliography{main}

\pagebreak

\appendix

%--------------------------------------------------------
%--------------------------------------------------------
\section{Properties of FG and TL}

The translation rules from FG to the target language also induce a set of typing
rules for FG by simply omitting the translation part. These typing rules are slightly different
from FG's original typing rules~\cite{FeatherweightGo},
because the translation rules are not syntax-directed due the structural subtyping rule \Rule{td-sub}
defined in \Cref{f:type-directed-exp}.
We could integrate this rule via the other rules but this would make all
the rules harder to read. Hence, we prefer to have a separate rule \Rule{td-sub}.
Nevertheless, the typing rules induced by the translation are equivalent to FG's original typing rules.

\begin{lemma}\label{lem:fg-typing-equiv}
  Assume $P$ is an FG program. Then we have that
  $P$ is well-typed according to FG's original typing rules \cite[Figure 13]{FeatherweightGo} if, and only if,
  $\tdProgTrans{P}{\program}$ for some $\program$.
\end{lemma}

\begin{proof} %%[of Lemma~\ref{lem:fg-typing-equiv}]
  We write
  $\fgEnv \turnsFG e : t$ for FG's original typing relation on expressions, and $\turnsFG P~\mathsf{ok}$
  for FG's original typing relation on programs. These relation are defined in \cite[Figure 13]{FeatherweightGo},
  the sequence of declaration $\foreachN{D}$ is implicit there. Further, we need
  to prepend $\mbox{\kw{package}}\,\mbox{\kw{main}};$ to program $P$.
  FG's original subtyping relation, written
  $\foreachN{D} \turnsFG \subtypeOf{t}{u}$ is identical to the one defined in \Cref{f:fg}, expect
  that in FG's original definition $\foreachN{D}$ is implicit.

  We then prove the following facts, where (e) is the claim of the lemma.
  \begin{EnumAlph}
  \item If $\foreachN{D} \turnsFG \subtypeOf{t}{u}$ then either $\tdUpcast{\foreachN{D}}{\subtypeOf{t}{u}}{\expT}$ for some $E$
    or $t = u$ and $t$ is a struct type.
  \item If $\tdUpcast{\foreachN{D}}{\subtypeOf{t}{u}}{\expT}$ then $\foreachN{D} \turnsFG \subtypeOf{t}{u}$.
  \item If $\fgEnv \turnsFG e : t$ then
    $\tdExpTrans{\pair{\foreachN{D}}{\fgEnv}}{e : t}{\expT}$ for some $\expT$.
  \item If $\tdExpTrans{\pair{\foreachN{D}}{\fgEnv}}{e : t}{\expT}$ then
    $\fgEnv \turnsFG e : t'$ for some $t'$ with $\foreachN{D} \turnsFG \subtypeOf{t'}{t}$.
  \item $\turnsFG P~\mathsf{ok}$ iff
    $\tdProgTrans{P}{\program}$ for some $\program$.
  \end{EnumAlph}

  We prove (a) and (b) by case distinctions on the last rule of the given derivations;
  (c) and (d) follow by induction on the derivations, using (a) and (b). Claim (e) then
  follows by examining the type rules, using (c), (d), and conditions FG2, FG3, FG4. \qed
\end{proof}

%% MS: note, the below results are used implicitly, see ``observing reductions''.
%% \ms{nuke, there's no reference to these results in any of the proofs,
%%   they may be necessary but not sufficient in our reasoning,
%%   e.g. we make observations about how to decompose a ``big'' evaluation step into some ``smaller'' steps,
%%   likely being deterministic is a consequence but doesn't seem to be directly needed.}

\begin{lemma}[FG reductions are deterministic]
  \label{lem:fg-determ}
  If $\foreachN{D} \turnsFG \reduce{e}{e'}$ and
  $\foreachN{D} \turnsFG \reduce{e}{e''}$ then $e = e'$.
\end{lemma}
\begin{proof} %%[of Lemma~\ref{lem:fg-determ}]
  We first state and prove three sublemmas:
  \begin{EnumAlph}
  \item If $e = \EvCtx_1[\EvCtx_2[e']]$ then there exists $\EvCtx_3$ with $e = \EvCtx_3[e']$.
    The proof is by induction on $\EvCtx_1$.
  \item If $\foreachN{D} \turnsFG \reduce{e}{e'}$ then there exists a derivation of
    $\foreachN{D} \turnsFG \reduce{e}{e'}$ that ends with
    at most one consecutive application of rule \Rule{fg-context}. The proof is by induction
    on the derivation of $\foreachN{D} \turnsFG \reduce{e}{e'}$. From the IH, we know that this derivation ends with
    at most \emph{two} consecutive applications of rule \Rule{fg-context}. If there are
    two such consecutive applications, (a) allow us to merge the two
    evaluation contexts involved, so that we need only one consecutive application
    of \Rule{fg-context}.
  \item We call an FGG expression \emph{directly reducible} if it reduces but not by rule \Rule{fg-context}.
    If $e_1$ and $e_2$ are now directly reducible and $\EvCtx_1[e_1] = \EvCtx_2[e_2]$ then $\EvCtx_1 = \EvCtx_2$
    and $e_1 = e_2$. For the proof, we first note that $\EvCtx_1 = \Hole$ iff $\EvCtx_2 = \Hole$. This
    holds because directly reducible expressions have no inner redexes. The rest of the
    proof is then a straightforward induction on $\EvCtx_1$.
  \end{EnumAlph}
  Now assume $\foreachN{D} \turnsFG \reduce{e}{e'}$ and $\foreachN{D} \turnsFG \reduce{e}{e''}$.
  By (b) we may assume that both derivations
  ends with at most one consecutive application of rule \Rule{fg-context}. It is easy to see
    (as values do not reduce) that both derivations must end with the same rule. If this rule
    is not \Rule{fg-context}, then obviously $e' = e''$ (note condition FG4 for rule \Rule{fg-call}). Otherwise,
    we have the following situation with $R_1 \neq \Rule{fg-context}$ and
    $R_2 \neq \Rule{fg-context}$:
    \begin{mathpar}
      \inferrule*[left=fg-context]{
        \inferrule*[left=$R_1$]{
        }{
          \reduce{e_1}{e_1'}
        }
      }{
        \reduce{\BraceBelow{\EvCtx_1[e_1]}{= e}}{\BraceBelow{\EvCtx_1[e_1']}{= e'}}
      }

      \inferrule*[right=fg-context]{
        \inferrule*[right=$R_2$]{
        }{
          \reduce{e_2}{e_2'}
        }
      }{
        \reduce{\BraceBelow{\EvCtx_2[e_2]}{= e}}{\BraceBelow{\EvCtx_2[e_2']}{= e''}}
      }
    \end{mathpar}
    As neither $R_1$ nor $R_2$ are \Rule{fg-context}, we know that $e_1$ and $e_2$ are directly
    reducible. Thus, with $\EvCtx_1[e_1] = \EvCtx_2[e_2]$ and (c) we get $\EvCtx_1 = \EvCtx_2$ and
    $e_1 = e_2$. With $R_1$ and $R_2$ not being \Rule{fg-context}, we have $e_1' = e_2'$, so
    $e' = e''$ as required.
    \qed
\end{proof}

\begin{lemma}[Target reductions are deterministic]
  \label{lem:tl-determ}
  If $\reduceTL{\vbMethTL}{\expT}{E'}$ and
  $\reduceTL{\vbMethTL}{\expT}{E''}$ then $E' = E''$.
  Further, if $\turnsTL \reduce{\program}{\program'}$ and
  $\turnsTL \reduce{\program}{\program''}$ then $\program = \program'$.
\end{lemma}

\begin{proof} %% [of Lemma~\ref{lem:tl-determ}]
  Assume $\reduceTL{\vbMethTL}{\expT}{E'}$ and
  $\reduceTL{\vbMethTL}{\expT}{E''}$.
  The claim that $E' = E''$ follows analogously to the proof of \Cref{lem:fg-determ}.
  If both derivations end with rule \Rule{tl-case}, we get deterministic evaluation
  by the syntactic restriction that case clauses have distinct constructors.

  Deterministic evaluation for programs is a simple consequence of deterministic
  evaluation of expressions. \qed
\end{proof}

% --------------------------------------------------------
% --------------------------------------------------------
\section{Logical Relation Properties}

A fundamental property of step indexed logical relations is that
if two expressions are in a relation for $k$ steps then they are also in a relation
for any smaller number of steps.
(See \Rule{lr-mono} from \Cref{f:step-index}.)

\begin{lemma}[Monotonicity]
\label{le:monotonicity}
  (1) Let $\redLRConfk{k}{t}{e}{\expT}$ and $\kA \leq k$.
Then, we find that $\redLRConfk{\kA}{t}{e}{\expT}$.
  (2) Let $\redVLRConfk{k}{t}{v}{\uT}$ and $\kA \leq k$.
  Then, we find that $\redVLRConfk{\kA}{t}{v}{\uT}$.
\end{lemma}

\begin{proof}
  By mutual induction over the derivations $\redLRConfk{k}{t}{e}{\expT}$
  and
  $\redVLRConfk{k}{t}{v}{\uT}$.

  \noindent
  {\bf Case} \Rule{red-rel-exp}: Follows immediately.

  \noindent
      {\bf Case} \Rule{red-rel-struct}:
\begin{mathpar}
 \inferrule
     { \TYPE\ t_S \ \STRUCT\ \{ \foreach{f \ t}{n} \} \in  \foreachN{D}
       %% \fieldsOf{\foreachN{D}}{t_S} = \foreach{f_i \ t_i}{n}
       \\ \forall i \in [n]. \redVLRConfk{k}{t_i}{v_i}{\uT_i}
     }
     { \redVLRConfk{k}
       {t_S}
       {t_S \{ \foreach{v}{n} \}}
       {\kT_{t_S} \ (\foreach{\uT}{n})}
     }
 \end{mathpar}

  Follows immediately by induction.

  \noindent
      {\bf Case} \Rule{red-rel-iface}:
\begin{mathpar}
 \inferrule
    {\uT = \kT_{u_S} \ \foreachN{\uTA}
      \\ (1)~ \forall k_1 \leq k . \redVLRConfk{k_1}{u_S}{v}{\uT}
      \\ \methodSpecifications{\foreachN{D}}{t_I} = \{ \foreach{m M}{n} \}
      \\ (2)~ \forall k_2 \leq k, i \in [n] . \redLRConfk{k_2}
                       {m_i M_i}
                       {\methodLookup{\foreachN{D}}{(m_i,u_S)}}
                       {\yT_i}
     }
     { \redVLRConfk{k}{t_I}{v}
             {\kT_{t_I} \ (\uT, \foreach{\yT}{n})}
     }
\end{mathpar}

  Consider the first premise (1).
  If there exists $k_1 \leq \kA$ then $\redVLRConfk{k_1}{u_S}{v}{\uT}$. Otherwise,
  this premise holds vacuously.
  \\
  The same argument for $k_2 \leq \kA$ applies to the second premise (2)
  by unfolding $\redLRConfk{k_2}
                       {m_i M_i}
                       {\methodLookup{\foreachN{D}}{(m_i,u_S)}}
                       {\yT_i}$ via rule \Rule{red-rel-method}.
  \\
  Hence,
  $\redVLRConfk{\kA}{t_I}{v}{\kT_{t_I} \ (\uT, \foreach{\yT}{n})}$.
  \qed
\end{proof}
A similar monotonicity result applies to method definitions and declarations.

\begin{lemma}[Monotonicity 2]
\label{le:monotonicity2}
Let $\redLRConfk{k}{m M}{\FUNC\ (x \ t_S) \ mM \ \{ \RETURN\ e \}}{\uT}$ and $\kA \leq k$.
Then, we find that $\redLRConfk{\kA}{m M}{\FUNC\ (x \ t_S) \ mM \ \{ \RETURN\ e \}}{\uT}$.
\end{lemma}
\begin{proof}
  Follows immediately by observing the premise of rule \Rule{red-rel-method}.
  \qed
\end{proof}

\begin{lemma}[Monotonicity 3]
\label{le:monotonicity3}
Let $\redLRk{k}{}{}{\foreachN{D}}{\vbMethTL}$ and $\kA \leq k$.
  Then, we find that $\redLRk{\kA}{}{}{\foreachN{D}}{\vbMethTL}$.
\end{lemma}
\begin{proof}
  Follows via Lemma~\ref{le:monotonicity2}.
  \qed
\end{proof}

Monotonicity is an essential property that is exploited frequently in our proofs.
Another useful property is \Rule{lr-step} from Figure~\ref{f:step-index}.
We also need several variations of this property in our proofs.

\begin{lemma}
 \label{le:source-step-preserve-equivalence}
  Let $\redLRConfk{k}{t}{e}{\expT}$ for some $k$, $e$, $\expT$, $\foreachN{D}$ and $\vbMethTL$.
  Let $\reduceFGk{1}{\foreachN{D}}{e_2}{e}$ for some $e_2$.
  Then, we have that $\redLRConfk{k+1}{t}{e_2}{\expT}$.
\end{lemma}
\begin{proof}
  Based on our assumption we find that
\begin{mathpar}
 \inferrule[red-rel-exp]{
   {
     \ba{c}
     (1) \left(
       \forall k' < k,
       v ~.~ \reduceFGk{\kA}{\foreachN{D}}{e}{v} \implies \exists \uT.
       \reduceTLN{\vbMethTL}{\expT}{\uT} \wedge \redVLRConfk{k- \kA}{t}{v}{\uT}
     \right)
     \\
     \wedge
     \\
     (2) \left(
       \forall k' < k,
       e' ~.~ \reduceFGk{k'}{\foreachN{D}}{e}{e'} \wedge \divergeDFG{e'}
       \implies
       \divergeMTL{\expT}
       \right)
    \PANIC{
     \\
     \wedge
     \\
     (3) \left(
       \forall k' < k,
       e' ~.~ \reduceFGk{\kA}{\foreachN{D}}{e}{e'} \wedge \panicFG{\foreachN{D}}{e'} \implies
              \panicMTL{\expT}
       \right)
    }
    %% PANIC
     \ea
   }
 }{
   \redLRConfk{k}{t}{e}{\expT}
 }
\end{mathpar}

From (1) and $\reduceFGk{1}{\foreachN{D}}{e_2}{e}$ we conclude
that
\bda{c}
\forall k' < k+1,
       v ~.~ \reduceFGk{\kA+1}{\foreachN{D}}{e_2}{v} \implies \exists \uT.
       \reduceTLN{\vbMethTL}{\expT}{\uT} \wedge \redVLRConfk{k + 1 - \kA}{t}{v}{\uT}
       \eda

From (2) and $\reduceFGk{1}{\foreachN{D}}{e_2}{e}$ we conclude
that
\bda{c}
       \forall k' < k+1,
       e' ~.~ \reduceFGk{k'}{\foreachN{D}}{e_2}{e'} \wedge \divergeDFG{e'}
       \implies
       \divergeMTL{\expT}
       \eda

\PANIC{
From (3) and $\reduceFGk{1}{\foreachN{D}}{e_2}{e}$ we conclude
that
\bda{c}
       \forall k' < k+1,
       e' ~.~ \reduceFGk{\kA}{\foreachN{D}}{e_2}{e'} \wedge \panicFG{\foreachN{D}}{e'} \implies
         \panicMTL{\expT}
\eda
}
%% PANIC

Thus, $\redLRConfk{k+1}{t}{e_2}{\expT}$ and we are done.
  \qed
\end{proof}

\begin{lemma}
 \label{le:target-step-preserve-equivalence}
  Let $\redLRConfk{k}{t}{e}{\expT}$ for some $k$, $e$, $\expT$, $\foreachN{D}$ and $\vbMethTL$.
  Let $\reduceTLk{k_2}{\vbMethTL}{\expT_2}{\expT}$ for some $\expT_2$ and $k_2$.
  Then, we have that $\redLRConfk{k}{t}{e}{\expT_2}$.
\end{lemma}
\begin{proof}
  By assumption we have the following
\begin{mathpar}
 \inferrule[red-rel-exp]{
   {
     \ba{c}
     \left(
       \forall k' < k,
       v ~.~ \reduceFGk{\kA}{\foreachN{D}}{e}{v} \implies \exists \uT.
       \reduceTLN{\vbMethTL}{\expT}{\uT} \wedge \redVLRConfk{k- \kA}{t}{v}{\uT}
     \right)
     \\
     \wedge
     \\
     \left(
       \forall k' < k,
       e' ~.~ \reduceFGk{k'}{\foreachN{D}}{e}{e'} \wedge \divergeDFG{e'}
       \implies
       \divergeMTL{\expT}
       \right)
     \PANIC{
     \\
     \wedge
     \\
     \left(
       \forall k' < k,
       e' ~.~ \reduceFGk{\kA}{\foreachN{D}}{e}{e'} \wedge \panicFG{\foreachN{D}}{e'} \implies
                \panicMTL{\expT}
       \right)
     }
     %% PANIC
     \ea
   }
 }{
   \redLRConfk{k}{t}{e}{\expT}
 }
\end{mathpar}

For each case we can argue that $\expT_2$ satisfies
the requirements. Hence, we find that  $\redLRConfk{k}{t}{e}{\expT_2}$.
  \qed
\end{proof}

\begin{lemma}
 \label{le:two-target-step-preserve-equivalence}
 Let $\redLRConfk{k}{t}{e}{\expT}$ for some $k$, $e$, $\expT$, $\foreachN{D}$ and $\vbMethTL$.
  Let $\reduceFGN{\foreachN{D}}{e}{v}$ for some value $v$.
  Let $\reduceTLN{\vbMethTL}{\expT}{\uT}$ for some value $\uT$.
  Let $\expTA$ be a target expression such that $\reduceTLN{\vbMethTL}{\expTA}{\uT}$.
  Then, we have that $\redLRConfk{k}{t}{e}{\expTA}$.
\end{lemma}
\begin{proof}
  Expression $e$ reduces to a value.
  Hence, the statement $\redLRConfk{k}{t}{e}{\expT}$ implies

\bda{c}
       \forall k' < k,
       v ~.~ \reduceFGk{\kA}{\foreachN{D}}{e}{v} \implies \exists \uT.
       \reduceTLN{\vbMethTL}{\expT}{\uT} \wedge \redVLRConfk{k- \kA}{t}{v}{\uT}
\eda
We can simply replace $\expT$ by $\expTA$ in the above.
Hence, we also find that $\redLRConfk{k}{t}{e}{\expTA}$.
\qed
\end{proof}

%---------------------------------------------------------------------------------
%---------------------------------------------------------------------------------
\section{Semantic Preservation for Interface-Value Constructors and Destructors}

Interface-value constructors and destructors, see Figure~\ref{f:upcast-downcast},
preserve equivalent expressions via logical relations as stated by the following results.

\begin{lemma}[Structural Subtyping versus Interface-Value Constructors]
\label{le:upcast-red-equiv}
  Let $\tdUpcast{\foreachN{D}}{\subtypeOf{t}{u}}{\expT_1}$
  and $\redLRk{k}{}{}{\foreachN{D}}{\vbMethTL}$
  and $\redLRConfk{k}{t}{e}{\expT_2}$.
  Then, we find that $\redLRConfk{k}{u}{e}{\expT_1 \ \expT_2}$.
\end{lemma}

\begin{proof}
We perform a case analysis of the derivation for $\tdUpcast{\foreachN{D}}{\subtypeOf{t}{u}}{\expT_1}$
and label the assumptions
(1) $\redLRk{k}{}{}{\foreachN{D}}{\vbMethTL}$
and (2) $\redLRConfk{k}{t}{e}{\expT_2}$
as well as the to be proven statement
(3) $\redLRConfk{k}{t}{e}{\expT_1 \ \expT_2}$ for later reference.

  \noindent
      {\bf Case} \Rule{td-cons-struct-iface}:
\begin{mathpar}
\inferrule
          {\TYPE\ t_I \ \INTERFACE\ \{ \foreachN{S} \} \in \foreachN{D}
        \\ \methodSpecifications{\foreachN{D}}{t_S} \supseteq \foreachN{S}
        \\ \foreachN{S} = \foreach{m M}{n}
          }
          {\tdUpcast{\foreachN{D}}{\subtypeOf{t_S}{t_I}}
                    {\lambda \xT.
                      \kT_{t_I} \ (\xT, \foreach{\mT{m_i}{t_S}}{n}) }
          }
\end{mathpar}

  We establish statement (3) by distinguishing among the subcases that arise in rule \Rule{red-rel-exp}.

  \addSpace
  {\bf Subcase-Terminate:}

      Suppose (4) $\reduceFGk{\kA}{\foreachN{D}}{e}{v}$
  for some $\kA$ and value $v$ where $\kA < k$.

  (5) $\reduceTLN{\vbMethTL}{\expT_2}{\uT}$ for some $\uT$ where

  (6) $\redVLRConfk{k - \kA}{t_S}{v}{\uT}$

  via reverse application of rule \Rule{red-rel-exp} on (2) where
  in the premise the left-hand side of the implication is satisfied via (4).

  (7) $\reduceTLN{\vbMethTL}{(\lambda \xT.  \kT_{t_I} \ (\xT, \foreach{\mT{m_i}{t_S}}{n})) \ \expT_2}{\kT_{t_I} \ (\uT, \foreach{\mT{m_i}{t_S}}{n})}$

  via reduction step (5).

  (8) $\redLRConfk{k}
                    {m_i M_i}
                    {\FUNC\ (x \ t_S) \ m_i M_i \ \{ \RETURN\ e \}}{\mT{m_i}{t_S}}$ for $i=1,...,n$

  via reverse application of rule \Rule{red-rel-decls} on (1).

   (9) $\redLRConfk{k - \kA}
                    {m_i M_i}
                    {\FUNC\ (x \ t_S) \ m_i M_i \ \{ \RETURN\ e \}}{\mT{m_i}{t_S}}$

   via (7) and the Monotonicity Lemma.

  (10) $\redVLRConfk{k - \kA}{t_I}{v}{\kT_{t_I} \ (\uT, \foreach{\mT{m_i}{t_S}}{n})}$

   via application of rule \Rule{red-rel-iface}
   on (6) and (8). Statements (6) and (8) hold for any $\kB \leq k - \kB$ via the Monotonicity Lemma.

      Thus, this subcase in statement (3) holds.

  \addSpace
  {\bf Subcase-Diverge:}

      Suppose (4) $\reduceFGk{\kA}{\foreachN{D}}{e}{e'}$
  for some $\kA$ and $e'$ where $\kA < k$ and $\divergeDFG{e'}$.

  (5) $\divergeMTL{\expT_2}$

    via reverse application of rule \Rule{red-rel-exp} on (2) where
  in the premise the left-hand side of the implication is satisfied via (4).

  (6) $\divergeMTL{(\lambda \xT.  \kT_{t_I} \ (\xT, \foreach{\mT{m_i}{t_S}}{n})) \ \expT_2}$ via (5).

  Thus, this subcase in statement (3) holds.

  \PANIC{
      \addSpace
    {\bf Subcase-Panic:}

    Suppose (4) $\reduceFGk{\kA}{\foreachN{D}}{e}{e'}$
    for some $\kA$ and $e'$ where $\kA < k$ and $\panicDFG{e'}$.

    (5) $\panicMTL{\expT_2}$

  via reverse application of rule \Rule{red-rel-exp} on (2) where
  in the premise the left-hand side of the implication is satisfied via (4).

    (6) $\panicMTL{(\lambda \xT.  \kT_{t_I} \ (\xT, \foreach{\mT{m_i}{t_S}}{n})) \ \expT_2}$
    via (5).

      Thus, this subcase in statement (3) holds.
}
%% PANIC

  \noindent
      {\bf Case} \Rule{td-cons-iface-iface}
\begin{mathpar}
\inferrule
          {\TYPE\ t_I \ \INTERFACE\ \{ \foreach{R}{n} \} \in \foreachN{D}
       \\ \TYPE\ u_I \ \INTERFACE\ \{ \foreach{S}{q} \} \in \foreachN{D}
       \\  S_i = R_{\mapPerm(i)} \quad\noteForall{i \in [q]}
        }
          {\tdUpcast{\foreachN{D}}{\subtypeOf{t_I}{u_I}}
                    { \lambda \xT. \CASE\,\xT\,\OF\
                          \kT_{t_I} \ (\xT, \foreach{\xT}{n})
                          \rightarrow
                          \kT_{u_I} \ (\xT, \xT_{\mapPerm(1)}, \ldots, \xT_{\mapPerm(q)})
                    }
         }
\end{mathpar}

Via similar reasoning as for the other upcast case,
we establish statement (3) by distinguishing among the subcases that arise in rule \Rule{red-rel-exp}.

  \addSpace
   {\bf Subcase-Terminate:}

   Suppose (4) $\reduceFGk{\kA}{\foreachN{D}}{e}{v}$
  for some $\kA$ and value $v$ where $\kA < k$.

  (5) $\reduceTLN{\vbMethTL}{\expT_2}{\uT}$ for some $\uT$ where

  (6) $\redVLRConfk{k - \kA}{t_I}{v}{\uT}$

  via reverse application of rule \Rule{red-rel-exp} on (2) where
  in the premise the left-hand side of the implication is satisfied via (4).

  (7) for any $\kB \leq k - \kA$

  (8) $\redVLRConfk{\kB}{u_S}{v}{\uTVal}$ and

  (9) $\redLRConfk{\kB}{R_j}{\FUNC \ (x \ u_S) \ R_j}{\yT_j}$ for $j=1,...,n$ where

  (10) $\uTVal = \kT_{u_S} \ (\foreachN{\uTA})$ and

  (11) $\uT = \kT_{t_I} \ (\uTVal, \foreach{\yT_j}{n})$

  via reverse application of rule \Rule{red-rel-iface} on (6).

 (12) $\reduceTLN{\vbMethTL}{(\lambda \xT. \CASE\,\xT\,\OF\
                          \kT_{t_I} \ (\xT, \foreach{\xT}{n})
                          \rightarrow
                          \kT_{u_I} \ (\xT, \xT_{\mapPerm(1)}, \ldots, \xT_{\mapPerm(q)})) \ \expT_2}
  {\kT_{u_I} \ (\uTVal, \yT_{\mapPerm(1)}, \ldots, \yT_{\mapPerm(q)})}$ via

  reduction steps (5), (10) and (11).

  (8) $\redVLRConfk{k - \kA}{u_I}{v}{\kT_{u_I} \ (\uTVal, \yT_{\mapPerm(1)}, \ldots, \yT_{\mapPerm(q)})}$

  via application of rule \Rule{red-rel-iface} on (6), (7) and (9) in combination
  with the Monotonicity Lemma.

      Thus, the first subcase in statement (3) holds.

      We omit the other
      subcase\PANIC{s} as the reasoning steps exactly correspond
      to case \Rule{td-cons-struct-iface}.
\qed
\end{proof}

\begin{lemma}[Type Assertions versus Interface-Value Destructors]
\label{le:downcast-red-equiv}
  Let $\tdDowncast{\foreachN{D}}{\assertOf{t}{u}}{\expT_1}$
  and $\redLRk{k}{}{}{\foreachN{D}}{\vbMethTL}$
  and $\redLRConfk{k}{t}{e}{\expT_2}$.
  Then, we find that $\redLRConfk{k}{u}{e.(u)}{\expT_1 \ \expT_2}$.
\end{lemma}

\begin{proof}
  We perform a case analysis of the derivation $\tdDowncast{\foreachN{D}}{\assertOf{t}{u}}{\expT_1}$
  and label the assumptions
(1) $\redLRk{k}{}{}{\foreachN{D}}{\vbMethTL}$
  and (2) $\redLRConfk{k}{t}{e}{\expT_2}$
  as well as the to be proven statement
(3) $\redLRConfk{k}{t}{e.(u)}{\expT_1 \ \expT_2}$ for later reference.

  \noindent
  {\bf Case} \Rule{td-destr-iface-struct}:
  \begin{mathpar}
  \inferrule
            {(4) \ \TYPE\ t_I \ \INTERFACE\ \{ \foreach{S}{n} \} \in \foreachN{D}
             \\ (5) \ \fgSub{\foreachN{D}}{t_S}{t_I}
            }
            {\tdDowncast{\foreachN{D}}{\assertOf{t_I}{t_S}}{\lambda \xT.  \CASE\,\xT\,\OF\
                    \kT_{t_I} \ (K_{t_S} \ \foreachN{\yT}, \foreach{\xT}{n})
                      \rightarrow K_{t_S} \ \foreachN{\yT}
                   }
            }
  \end{mathpar}

  We establish statement (3) by distinguishing among the subcases that arise in rule \Rule{red-rel-exp}.

  \addSpace
  {\bf Subcase-Terminate:}
   Suppose (6) $\reduceFGk{\kA}{\foreachN{D}}{e.(t_S)}{v}$  for some value $v$ where $\kA < k$.

   (7) $v = t_S \{ \foreachN{v} \}$ for some $\foreachN{v}$

   via (5), (6) and the FG reduction rules.

   (8) $\reduceFGk{\kA - 1}{\foreachN{D}}{e}{t_S \{ \foreachN{v} \}}$

  via (6) and (7).

  (9) $\reduceTLN{\vbMethTL}{\expT_2}{\uT}$ for some $\uT$ where

  (10) $\redVLRConfk{k - (\kA - 1)}{t_I}{v}{\uT}$

  via reverse application of rule \Rule{red-rel-exp} on (2)
  where in the premise the right-hand side of the implication is satisfied via (8)
  and the fact that $\kA - 1 < k$.

  (11) $\redVLRConfk{\kB}{t_S}{v}{\uTVal}$  for any $\kB \leq k - (\kA - 1)$ where

  (12) $\uTVal = \kT_{t_S} \ (\foreachN{\uTA})$ for some $\foreachN{\uTA}$ and

  (13) $\uT = \kT_{t_I} \ (\uTVal, \foreach{\yT_j}{n})$ for some $\foreach{\yT_j}{n}$

  via reverse application of rule \Rule{red-rel-iface} on (10)
  where we make use of (7) and rule \Rule{red-rel-struct} to derive the shape of $\uTVal$ and (4) to guarantee
  that there are $n$ method variables $\yT_j$.

  (14) $\reduceTLN{\vbMethTL}{(\lambda \xT.  \CASE\,\xT\,\OF\
                    \kT_{t_I} \ (K_{t_S} \ \foreachN{\yT}, \foreach{\xT}{n})
                    \rightarrow K_{t_S} \ \foreachN{\yT}) \ \expT_2}{\uTVal}$

  via the reduction step (9) and the equations (12) and (13).

  (15) $\redVLRConfk{k - \kA}{t_S}{v}{\uTVal}$

  via (11) and the Monotonicity Lemma.

  Thus, we can establish this subcase in statement (3).

  \addSpace
  {\bf Subcase-Diverge:}
  Suppose (6) $\reduceFGk{\kA}{\foreachN{D}}{e.(t_S)}{e'}$  for some $e'$ where $\kA < k$ and $\divergeDFG{e'}$.

  (7) $\reduceFGk{\kA}{\foreachN{D}}{e}{e''}$  for some $e''$ where $\kB \leq kA$ and $\divergeDFG{e''}$

   by observing the reduction (6).

  (8) $\divergeMTL{\expT_2}$

    via reverse application of rule \Rule{red-rel-exp} on (2) where
    in the premise the left-hand side of the implication is satisfied via (8).

  (9) $\divergeMTL{(\lambda \xT.  \CASE\,\xT\,\OF\
                    \kT_{t_I} \ (K_{t_S} \ \foreachN{\yT}, \foreach{\xT}{n})
                    \rightarrow K_{t_S} \ \foreachN{\yT}) \expT_2}$ via (7).

   Thus, we can establish this subcase in statement (3).

\PANIC{
  \addSpace
  {\bf Subcase-Panic:}
  Suppose (6) $\reduceFGk{\kA}{\foreachN{D}}{e.(t_S)}{e'}$  for some $e'$ where $\kA < k$ and $\panicDFG{e'}$.

  We distinguish among the following two cases.
  Either
  (1) the expression panics or
  (2) the type assertion fails.

  \addSpace
  {\bf Subcase-Panic-1:}
  (7) $\reduceFGk{\kB}{\foreachN{D}}{e}{e''}$  for some $e''$ where $\kB < k$ and $\panicDFG{e'}$.

  \addSpace
  {\bf Subcase-Panic-2:}
  (8) $\reduceFGk{\kB}{\foreachN{D}}{e}{u_S \{ \foreachN{v'} \}}$
for some $u_S \{ \foreachN{v'} \}$ where $\kB < k$ and $t_S \not = u_S$.

  \addSpace
  Consider {\bf Subcase-Panic-1.}

  (9) $\panicMTL{\expT_2}$

  via reverse application of rule \Rule{red-rel-exp} on (2) where
  in the premise the left-hand side of the implication is satisfied via (7).

  (10) $\panicMTL{(\lambda \xT.  \CASE\,\xT\,\OF\
                    \kT_{t_I} \ (K_{t_S} \ \foreachN{\yT}, \foreach{\xT}{n})
                    \rightarrow K_{t_S} \ \foreachN{\yT}) \ \expT_2}$ via (9) and we are done here.

  \addSpace
  Consider {\bf Subcase-Panic-2.}

  (11) $\reduceTLN{\vbMethTL}{\expT_2}{\uT}$ for some $\uT$ where

  (12) $\redVLRConfk{k - \kB}{t_I}{u_S \{ \foreachN{v'} \}}{\uT}$

  via reverse application of rule \Rule{red-rel-exp} on (2)
  where in the premise the right-hand side of the implication is satisfied via (8)

  (13) $\uTVal = \kT_{u_S} \ (\foreachN{\uTA})$ for some $\foreachN{\uTA}$ and

  (14) $\uT = \kT_{t_I} \ (\uTVal, \foreachN{\yT})$ for some method variables $\foreachN{\yT}$

  via reverse application of rule \Rule{red-rel-iface} on (12)
  where we make use of \Rule{red-rel-struct} to derive the shape of $\uTVal$.

  (15) $\panicMTL{(\lambda \xT.  \CASE\,\xT\,\OF\
                    \kT_{t_I} \ (K_{t_S} \ \foreachN{\yT}, \foreach{\xT}{n})
                    \rightarrow K_{t_S} \ \foreachN{\yT}) \ \expT_2}$

  via (11), (13) and (14) and the fact that $K_{t_S} \not = K_{u_S}$. Thus, we are done here.
} %% PANIC

   \noindent
   {\bf Case} \Rule{td-destr-iface-iface}:
   \begin{mathpar}
     \inferrule
     {
  (4) \ \TYPE\ t_I \ \INTERFACE\ \{ \foreach{R}{n} \} \in \foreachN{D}
            \\
            {
              (5) \ \textrm{for all}~\TYPE\ t_{Sj}\ \STRUCT\ \{ \foreachN{f \ u} \} \in \foreachN{D}
              ~\textrm{with}~ \tdUpcast{\foreachN{D}}{\subtypeOf{t_{Sj}}{u_I}}{\expT_j}
              \textrm{:}
            }
            \\
            {
              \clsT_j =
              \kT_{t_{Sj}} \ Y' \rightarrow (\expT_j \ (\kT_{t_{Sj}} \ Y'))
            }
          }
          { \tdDowncast{\foreachN{D}}{\assertOf{t_I}{u_I}}{\lambda \xT.
                             \CASE\ \xT \ \OF\
                             \kT_{t_I}\, (\yT, \foreach{\xT}{n}) \rightarrow \CASE\ \yT  \ \OF\
                             [\foreachN{\clsT}]}
          }
   \end{mathpar}

  We establish statement (3) by distinguishing among the subcases that arise in rule \Rule{red-rel-exp}.

  \addSpace
  {\bf Subcase-Terminate:}
   Suppose (6) $\reduceFGk{\kA}{\foreachN{D}}{e.(u_I)}{v}$  for some value $v$ where $\kA < k$.

   (7) $v = t_S \{ \foreachN{v} \}$ for some $\foreachN{v}$ where

   (8) $\fgSub{\foreachN{D}}{t_S}{u_I}$

   via (6) and the FG reduction rules.

   (9) $\reduceFGk{\kA - 1}{\foreachN{D}}{e}{t_S \{ \foreachN{v} \}}$

   via (6) and (7).

   (10) $\reduceTLN{\vbMethTL}{\expT_2}{\uT}$ for some $\uT$ where

   (11) $\redVLRConfk{k - (\kA - 1)}{t_I}{v}{\uT}$

   via reverse application of rule \Rule{red-rel-exp} on (2)
   where in the premise the right-hand side of the implication is satisfied via (9)
   and the fact that $\kA - 1 < k$.

   (12) $\redVLRConfk{\kB}{t_S}{v}{\uTVal}$  for any $\kB \leq k - (\kA - 1)$ where

  (12) $\uTVal = \kT_{t_S} \ (\foreachN{\uTA})$ for some $\foreachN{\uTA}$ and

  (13) $\uT = \kT_{t_I} \ (\uTVal, \foreach{\yT_j}{n})$ for some $\foreach{\yT_j}{n}$

  via reverse application of rule \Rule{red-rel-iface} on (11)
  where we make use of (7) and rule \Rule{red-rel-struct} to derive the shape of $\uTVal$ and (4) to guarantee
  that there are $n$ method variables $\yT_j$.

  (14) $\redLRConfk{k}{t_S}{e}{\uTVal}$

  via (9), (12) and rule \Rule{red-rel-exp}.
  %% MS: the FG reduction rules guarantee that the reduction leads to a value, neither panic nor diverge will arise.

  (15) $\redLRConfk{k}{u_I}{e}{\expTA \ \uTVal}$ where

  (16) $\tdUpcast{\foreachN{D}}{\subtypeOf{t_S}{u_I}}{\expTA}$

  via (8), (14) and Lemma~\ref{le:upcast-red-equiv}.

  (17) $\reduceTLN{\vbMethTL}{\expTA \ \uTVal}{\uTA}$ and

  (18) $\reduceTLN{\vbMethTL}{(\lambda \xT.
                             \CASE\ \xT \ \OF\
                             \kT_{t_I}\, (\yT, \foreach{\xT}{n}) \rightarrow \CASE\ \yT  \ \OF\
                                [\foreachN{\clsT}]) \ \expT_2}{\uTA}$ for some $\uTA$

  via (4), (10), (13) and (16).

  (19) $\redLRConfk{k}{u_I}{e}{(\lambda \xT.
                             \CASE\ \xT \ \OF\
                             \kT_{t_I}\, (\yT, \foreach{\xT}{n}) \rightarrow \CASE\ \yT  \ \OF\
                                [\foreachN{\clsT}]) \ \expT_2}$

  via (9), (17), (18) and Lemma~\ref{le:two-target-step-preserve-equivalence}.

  (20) $\redLRConfk{k}{u_I}{e.(u_I)}{(\lambda \xT.
                             \CASE\ \xT \ \OF\
                             \kT_{t_I}\, (\yT, \foreach{\xT}{n}) \rightarrow \CASE\ \yT  \ \OF\
                                [\foreachN{\clsT}]) \ \expT_2}$

  via (6), (7), (9), (19) and the Monotonicity Lemma.

   Thus, this subcase in statement (3) holds.

     \addSpace
  {\bf Subcase-Diverge:}
  Suppose (6) $\reduceFGk{\kA}{\foreachN{D}}{e.(u_I)}{e'}$  for some $e'$ where $\kA < k$ and $\divergeDFG{e'}$.

  (7) $\reduceFGk{\kA}{\foreachN{D}}{e}{e''}$  for some $e''$ where $\kB \leq kA$ and $\divergeDFG{e''}$

   by observing the reduction (6).

  (8) $\divergeMTL{\expT_2}$

    via reverse application of rule \Rule{red-rel-exp} on (2) where
    in the premise the left-hand side of the implication is satisfied via (8).

  (9) $\divergeMTL{(\lambda \xT.
                             \CASE\ \xT \ \OF\
                             \kT_{t_I}\, (\yT, \foreach{\xT}{n}) \rightarrow \CASE\ \yT  \ \OF\
                                [\foreachN{\clsT}]) \ \expT_2}$ via (7).

   Thus, we can establish this subcase in statement (3).

\PANIC{
  \addSpace
  {\bf Subcase-Panic:}
  Suppose (6) $\reduceFGk{\kA}{\foreachN{D}}{e.(u_I)}{e'}$  for some $e'$ where $\kA < k$ and $\panicDFG{e'}$.

  We distinguish among the following two cases.
  Either
  (1) the expression panics or
  (2) the type assertion fails.

  \addSpace
  {\bf Subcase-Panic-1:}
  (7) $\reduceFGk{\kB}{\foreachN{D}}{e}{e''}$  for some $e''$ where $\kB < k$ and $\panicDFG{e'}$.

  \addSpace
  {\bf Subcase-Panic-2:}
  (8) $\reduceFGk{\kB}{\foreachN{D}}{e}{u_S \{ \foreachN{v'} \}}$
for some $u_S \{ \foreachN{v'} \}$ where $\kB < k$ and $\fgSub{\foreachN{D}}{u_S}{u_I}$ does not hold.

  \addSpace
  Consider {\bf Subcase-Panic-1.}

  (9) $\panicMTL{\expT_2}$

  via reverse application of rule \Rule{red-rel-exp} on (2) where
  in the premise the left-hand side of the implication is satisfied via (7).

  (10) $\panicMTL{(\lambda \xT.
                             \CASE\ \xT \ \OF\
                             \kT_{t_I}\, (\yT, \foreach{\xT}{n}) \rightarrow \CASE\ \yT  \ \OF\
                                [\foreachN{\clsT}]) \ \expT_2}$ via (9) and we are done here.

  \addSpace
  Consider {\bf Subcase-Panic-2.}

  (11) $\reduceTLN{\vbMethTL}{\expT_2}{\uT}$ for some $\uT$ where

  (12) $\redVLRConfk{k - \kB}{t_I}{u_S \{ \foreachN{v'} \}}{\uT}$

  via reverse application of rule \Rule{red-rel-exp} on (2)
  where in the premise the right-hand side of the implication is satisfied via (8)

  (13) $\uTVal = \kT_{u_S} \ (\foreachN{\uTA})$ for some $\foreachN{\uTA}$ and

  (14) $\uT = \kT_{t_I} \ (\uTVal, \foreachN{\yT})$ for some method variables $\foreachN{\yT}$

  via reverse application of rule \Rule{red-rel-iface} on (12)
  where we make use of \Rule{red-rel-struct} to derive the shape of $\uTVal$.

  (15) $\panicMTL{(\lambda \xT.
                             \CASE\ \xT \ \OF\
                             \kT_{t_I}\, (\yT, \foreach{\xT}{n}) \rightarrow \CASE\ \yT  \ \OF\
                                [\foreachN{\clsT}]) \ \expT_2}$

  via the assumption that $\fgSub{\foreachN{D}}{u_S}{u_I}$ does not hold
  and therefore none of the pattern clauses $\foreachN{\clsT}$ will yield a match.
  Thus, we are done here.
} %% PANIC

      \qed
\end{proof}

%--------------------------------------------------------
%--------------------------------------------------------
%%\section{Proofs for Properties Stated in the Main Text}
%--------------------------------------------------------
\section{Proof of Lemma~\ref{le:exp-red-equivalent}}

\begin{proof}
  By induction over the derivation $\tdExpTrans{\pair{\foreachN{D}}{\fgEnv}}{e : t}{\expT}$.
  We label the assumptions   (1) $\redLRk{k}{\triple{\foreachN{D}}{\vbMethTL}{\fgEnv}}{}{\vbFG}{\vbTL}$
  and
  (2) $\redLRk{k}{}{}{\foreachN{D}}{\vbMethTL}$ as well as the to be proven
  statement (3) $\redLRConfk{k}{t}{\vbFG(e)}{\vbTL(\expT)}$ for some later reference.

  \noindent
  {\bf Case} \Rule{td-var}:
  \begin{mathpar}
    \inferrule
            { (x : t) \in \fgEnv
            }
            { \tdExpTrans{\pair{\foreachN{D}}{\fgEnv}}{x : t}{\xT}
            }
  \end{mathpar}

  (3) follows immediately from (1).

  \noindent
      {\bf Case} \Rule{td-struct}:
  \begin{mathpar}
    \inferrule
        {   \TYPE\ t_S \ \STRUCT\ \{ \foreachN{f_i \ t_i}{n} \} \in  \foreachN{D}
              %% \fieldsOf{\foreachN{D}}{t_S} = \foreach{f_i \ t_i}{n}
               \\   \tdExpTrans{\pair{\foreachN{D}}{\fgEnv}}{e_i : t_i}{\expT_i}\quad\noteForall{i \in [n]}
            }
            { \tdExpTrans{\pair{\foreachN{D}}{\fgEnv}}{t_S \{ \foreach{e_i}{n} \} : t_S }{\kT_{t_S} \ (\foreach{\expT_i}{n})}
            }
  \end{mathpar}

  (4) $\redLRConfk{k}{t}{\vbFG(e_i)}{\vbTL(\expT_i)}$ by induction.

  To establish (3) $\redLRConfk{k}{t}{\vbFG(e)}{\vbTL(\expT)}$
  we consider the subcases that arise in rule \Rule{red-rel-exp}.

    \addSpace
  {\bf Subcase-Terminate:}

  Suppose (5) $\reduceFGk{\kA}{\foreachN{D}}{\vbFG(t_S \{ \foreach{e_i}{n} \})}{t_S \{ \foreach{v_i}{n} \}}$
  for some values $v_i$ where  (6) $\kA < k$.

  (7) $\reduceFGk{k_i}{\foreachN{D}}{\vbFG(e_i)}{v_i}$ for some  $k_i$ where

  (8) $k_i \leq \kA < k$

  by observing the reduction (5) and the number of reduction steps taken (6).

  (9) $\reduceTLN{\vbMethTL}{\vbTL(\expT_i)}{\uT_i}$ for some $\uT_i$  where

  (10) $\redVLRConfk{k-k_i}{t_i}{v_i}{\uT_i}$

  via reverse application of rule \Rule{red-rel-exp} on (4)
  where in the premise the left-hand side of the implication is satisfied via (7) and (8).

  (11) $k- \kA \leq k - k_i$ via (8).

  (12) $\redVLRConfk{k-\kA}{t_i}{v_i}{\uT_i}$

  via (10), (11) and the Monotonicity Lemma.

  (10) $\reduceTLN{\vbMethTL}{\kT_{t_S} \ (\foreach{\vbTL(\expT_i)}{n})}{\kT_{t_S} \ (\foreach{\uT_i}{n})}$

  via the reduction step (7).

  (11) $\redVLRConfk{k-\kA}{t_S}{t_S \{ \foreach{e'_i}{n} \}}{\kT_{t_S} \ (\foreach{\uT_i}{n})}$

  via application of rule \Rule{red-rel-struct} on (9).

  Via (10) and (11) we can establish this subcase in statement (3).

    \addSpace
  {\bf Subcase-Diverge:}

    Suppose (5) $\reduceFGk{\kA}{\foreachN{D}}{\vbFG(t_S \{ \foreach{e_i}{n} \})}{t_S \{ \foreach{e'_i}{n} \}}$
  for some expressions $e'_1,...,e'_n$ where $\kA < k$
  and $\divergeDFG{t_S \{ \foreach{e'_i}{n} \}}$.

  (6) $\divergeDFG{e'_j}$ for some $j$ where $e'_1,...,e'_{j-1}$ are values via (5)
     and FG reduction rules.

  (7) $\reduceTLN{\vbMethTL}{\vbTL(\expT_l)}{\uT_l}$ for some values $\uT_l$
  for $l=1,...,j-1$ via (4), (6) and rule \Rule{red-rel-exp}.

  (8)  $\divergeMTL{\expT_j}$ via (4), (6) and rule \Rule{red-rel-exp}.

  (9) $\divergeMTL{\vbTL(\kT_{t_S} \ (\foreach{\expT_i}{n}))}$
       via (7) and (8).

  Via (9) we can establish this subcase in statement (3).

%% Identical reasoning as for ``diverge''.
  \PANIC{
    \addSpace
  {\bf Subcase-Panic:}

  Suppose (5) $\reduceFGk{\kA}{\foreachN{D}}{\vbFG(t_S \{ \foreach{e_i}{n} \})}{t_S \{ \foreach{e'_i}{n} \}}$
  for some  $e'_1,...,e'_n$ where $\kA < k$
  and  $\panicFG{\foreachN{D}}{t_S \{ \foreach{e'_i}{n} \}}$.

  (6) $\panicFG{\foreachN{D}}{e'_j}$ for some $j$ where $e'_1,...,e'_{j-1}$ are values via (5)
      and evaluation context $t_S\{e'_i,\ldots,e'_{j-1},\EvCtx,e_{j+1},\ldots,e_n\}$.

  (7) $\reduceTLN{\vbMethTL}{\vbTL(\expT_l)}{\uT_l}$ for some values $\uT_l$
  for $l=1,...,j-1$ via (4), (6) and rule \Rule{red-rel-exp}.

  (8)  $\panicMTL{\expT}$ via (4), (6) and rule \Rule{red-rel-exp}.

  (9) $\panicMTL{\vbTL(\kT_{t_S} \ (\foreach{\expT_i}{n}))}$
       via (7) and (8).

   Via (9) we can establish this subcase in statement (3).
}
%% PANIC

  \noindent
  {\bf Case} \Rule{td-access}:
  \begin{mathpar}
      \inferrule
           {  \tdExpTrans{\pair{\foreachN{D}}{\fgEnv}}{e : t_S}{\expT}
             \\  \TYPE\ t_S \ \STRUCT\ \{ \foreach{f_j \ t_j}{n} \} \in  \foreachN{D}
                 %% \fieldsOf{\foreachN{D}}{t_S} = \foreach{f_j \ t_j}{n}
          }
           { \tdExpTrans{\pair{\foreachN{D}}{\fgEnv}}
                        {e.f_i : t_i }
                        { \CASE\ \ \expT\ \ \OF\ \kT_{t_S} \ (\foreach{\xT_j}{n}) \rightarrow \xT_i}
           }
  \end{mathpar}

  (4) $\redLRConfk{k}{t}{\vbFG(e)}{\vbTL(\expT)}$ by induction.

  To establish (3) $\redLRConfk{k}{t}{\vbFG(e)}{\vbTL(\expT)}$
  we consider the subcases that arise in rule \Rule{red-rel-exp}.

    \addSpace
        {\bf Subcase-Terminate:}

  Suppose (5) $\reduceFGk{\kA}{\foreachN{D}}{\vbFG(e.f_i)}{v'}$
  for some value $v'$ where (6) $\kA < k$.

  (7) $\reduceFGk{\kB}{\foreachN{D}}{\vbFG(e)}{v}$  for some $v$ and $\kB$ where

  (8) $\kB < \kA$

  by observing the reduction (5) and the number of reduction steps taken (6).

  (9) $\reduceTLN{\vbMethTL}{\vbTL(\expT)}{\uT}$ for some $\uT$   where

  (10) $\redVLRConfk{k-\kB}{t_S}{v}{\uT}$

  via reverse application of rule \Rule{red-rel-exp} on (4)
  where the left-hand side of the implication is satisfied via (7) and (8).

  (11) $v = t_S \{ \foreach{v_j}{n} \}$ and

  (12) $\uT = \kT_{t_S} \ (\foreach{\uT_j}{n})$  for some $v_j$ and $\uT_j$ where

  (13) $\redVLRConfk{k-\kB}{t_j}{v_j}{\uT_j}$ for $j=1,...,n$

  via reverse application of rule \Rule{red-rel-struct} on (10).

  (14) $\reduceTLN{\vbMethTL}{\vbTL(\CASE\ \ \expT\ \ \OF\ \kT_{t_S} \ (\foreach{\xT_j}{n}) \rightarrow \xT_i)}{\uT_i}$

  via reduction step (9) and (12).

  (15) $v' = v_i$

  via reduction step (5) and (11).

  (16) $\redVLRConfk{k-\kA}{t_j}{v'}{\uT_i}$

  via (13), (15) and the Monotonicity Lemma   as we have that $k - \kA \leq k - \kB$.

  Via (14) and (16) we can establish this subcase in statement (3).

    \addSpace
        {\bf Subcase-Diverge:}

  Suppose (5) $\reduceFGk{\kA}{\foreachN{D}}{\vbFG(e.f_i)}{e'}$
  for some $e'$ where $\kA < k$ and $\divergeDFG{e'}$.

  (6) $\reduceFGk{k'}{\foreachN{D}}{\vbFG(e)}{e'}$
  and $\divergeDFG{e'}$ via (5) and the FG reduction rules.

  (7) $\divergeMTL{\vbTL(\expT)}$ via (4), (6) and rule \Rule{red-rel-exp}.

  (8) $\divergeMTL{\vbTL(\CASE\ \ \expT\ \ \OF\ \kT_{t_S} \ (\foreach{\xT_j}{n}) \rightarrow \xT_i)}$
  via (7).

  Via (8) we can establish this subcase in statement (3).

  %% Again almost identical reasoning.
  %% Step (6) is a bit more involved as we also need to consult the panic rules.
  \PANIC{
    \addSpace
        {\bf Subcase-Panic:}

  Suppose (5) $\reduceFGk{\kA}{\foreachN{D}}{\vbFG(e.f_i)}{e'}$
  for some $e'$ where $\kA < k$ and $\panicDFG{e'}$.

  (6) $\reduceFGk{k'}{\foreachN{D}}{\vbFG(e)}{e'}$
  and $\panicDFG{e'}$ via (5) and the FG reduction and panic rules.

  (7) $\panicMTL{\vbTL(\expT)}$ via (4), (6) and rule \Rule{red-rel-exp}.

  (8) $\panicMTL{\vbTL(\CASE\ \ \expT\ \ \OF\ \kT_{t_S} \ (\foreach{\xT_j}{n}) \rightarrow \xT_i)}$
  via (7).

  Via (8) we can establish this subcase in statement (3).
}
%% PANIC

    \noindent
  {\bf Case} \Rule{td-call-struct}:
  \begin{mathpar}
      \inferrule
            { m (\foreach{x_i \ t_i}{n}) \ t \in \methodSpecifications{\foreachN{D}}{t_S}
              \\
              \tdExpTrans{\pair{\foreachN{D}}{\fgEnv}}{e : t_S}{\expT}
              \\ \tdExpTrans{\pair{\foreachN{D}}{\fgEnv}}{e_i : t_i}{\expT_i}\quad\noteForall{i \in [n]}
          }
            { \tdExpTrans{\pair{\foreachN{D}}{\fgEnv}}{e.m(\foreach{e_i}{n}) : t}{\mT{m}{t_S} \ \expT \ (\foreach{\expT_i}{n}) } }
  \end{mathpar}

  (4) $\redLRConfk{k}{t_S}{\vbFG(e)}{\vbTL(\expT)}$
     and (5) $\redLRConfk{k}{t_i}{\vbFG(e_i)}{\vbTL(\expT_i)}$ by induction.

  To establish (3) $\redLRConfk{k}{t}{\vbFG(e)}{\vbTL(\expT)}$
  we consider the subcases that arise in rule \Rule{red-rel-exp}.

    \addSpace
        {\bf Subcase-Terminate:}

  Suppose (6) $\reduceFGk{\kA}{\foreachN{D}}{\vbFG(e.m(\foreach{e_i}{n}))}{v'}$
  for some value $v'$ where (7) $\kA < k$.

  (8) $\reduceFGk{\kB}{\foreachN{D}}{\vbFG(e)}{v}$ for some $v$, $\kB$ where

  (9) $\kB < \kA$

  by observing the reduction (6) and the number of steps taken (7).

  (10) $\reduceFGk{k_i}{\foreachN{D}}{\vbFG(e_i)}{v_i}$ for some $v_i$, $k_i$ where

  (11) $\sum_i k_i < \kA$

  by observing the reduction (6) and the number of steps taken (7).

  (12) $\reduceTLN{\vbMethTL}{\vbTL(\expT)}{\uT}$ for some $\uT$ where

  (13) $\redVLRConfk{k-\kB}{t_S}{v}{\uT}$

  via reverse application of rule \Rule{red-rel-exp} on (4) where
  in the premise the left-hand side of the implication is satisfied via (8) and (9).

  (14) $\reduceTLN{\vbMethTL}{\vbTL(\expT_i)}{\uT_i}$ for some $\uT_i$ where

  (15) $\redVLRConfk{k-k_i}{t_i}{v_i}{\uT_i}$

  via reverse application of rule \Rule{red-rel-exp} on (5) where
  in the premise the left-hand side of the implication is satisfied via (10) and (11).

  (16) Set $\kC = \minF(k-\kB,k - \sum_i k_i) - 1$ where

  (17) $0 \leq \kC < k$ and

  (18) $\kC < k- \kB$ and $\kC < k - k_i$

  via (7), (9) and (11).

  (19) $\redVLRConfk{\kC}{t_S}{v}{\uT}$

   via (13), (18) and the Monotonicity Lemma.

   (20) $\redVLRConfk{\kC}{t_i}{v_i}{\uT_i}$

   via (15), (18) and the Monotonicity Lemma.

   (21) $\redLRConfk{k}{m M}{\FUNC\ (x \ t_S) \ m M \{ \RETURN\ e'' \}}{\mT{m}{t_S}}$

   via reverse application of rule \Rule{red-rel-decls} on (2).

   (22) $\redLRConfk{\kC}{t}
         {\Angle{\subst{x}{v},\foreach{\subst{x_i}{v_i}}{n}} e''}
         {\mT{m}{t_S} \ \uT \ (\foreach{\uT_i}{n})}$

    via reverse application of rule \Rule{red-rel-method} on (21)
    where in the premise the left-hand side of the implication is satisfied
    via (18), (19) and (20).

    (23) $\redLRConfk{\kC+1}{t}
                     {v.m(\foreach{v_i}{n})}
                     {(\mT{m}{t_S} \ \uT) \ (\foreach{\uT_i}{n})}$

     via (22) and Lemma~\ref{le:source-step-preserve-equivalence}.

     (24) $\reduceFGk{\kA- (\kB + \sum_i k_i)}{\foreachN{D}}{v.m(\foreach{v_i}{n})}{v'}$

     via reduction steps (6), (8) and (10).

     (25) $ \kC + 1 > \kA - (\kB + \sum_i k_i)$

     via the following reasoning.

        \bda{llr}
         & \kC + 1
        \\ = & \minF(k-\kB,k - \sum_i k_i)     & \mbox{by definition}
        \\ = & k - \maxF(\kB, \sum_i k_i)      & \mbox{by $\minF/\maxF$ distributivity law}
        \\ \geq & k - (\kB + \sum_i k_i)       & \mbox{by $\maxF$ approximation}
        \\ > & \kA - (\kB + \sum_i k_i)        & \mbox{by $\kA < k$}
        \eda

    (26) $\reduceTLN{\vbMethTL}{\mT{m}{t_S} \ \uT \ (\foreach{\uT_i}{n})}{\uTA}$ for some $\uTA$ where

        (27) $\redVLRConfk{\kC+ 1 - (\kA- (\kB + \sum_i k_i))}{t}{v'}{\uTA}$

        via reverse application of rule \Rule{red-rel-exp} on (23)
        where in the premise the left-hand side of the implication is satisfied
        via (24) and (25).

   (28) $k - \kA \leq \kC+ 1 - (\kA- (\kB + \sum_i k_i))$

        via the following reasoning.

        \bda{llr}
        & \kC+ 1 - (\kA- (\kB + \sum_i k_i))
        \\ = &  \minF(k-\kB,k - \sum_i k_i)- (\kA- (\kB + \sum_i k_i))     & \mbox{by definition}
        \\ = & k - \maxF(\kB, \sum_i k_i)  - (\kA- (\kB + \sum_i k_i))     & \mbox{by $\minF/\maxF$ distributivity law}
        \\ = & k - \kA + \kB + \sum_i k_i  - \maxF(\kB, \sum_i k_i)        & \mbox{by equivalence}
        \\ \geq & k - \kA                                                 & \mbox{by approximation}
        \eda

        (29) $\redVLRConfk{k-\kA}{t}{v'}{\uTA}$

        via (27), (28) and the Monotonicity Lemma.

  (30) $\reduceTLN{\vbMethTL}{\vbTL(\mT{m}{t_S} \ \expT \ (\foreach{\expT_i}{n}))}{\uTA}$

        via reduction steps (12), (14) and (24).

        Via (29) and (30) we can establish this subcase in statement (3).

         \addSpace
        {\bf Subcase-Diverge:}
        Suppose (5) $\reduceFGk{k'}{\foreachN{D}}{\vbFG(e.m(\foreach{e_i}{n}))}{e'}$ for some $e'$
        where $\kA < k$ and $\divergeDFG{e'}$.

        We distinguish among the following three cases.
        Either the
        (1) expression on which the method call is performed diverges or
        (2) one of the arguments diverges or
        (3) reduction of the method call leads to some expression that diverges.

          \addSpace
        {\bf Subcase-Diverge-1:} $e' = e''.m(...)$ where
        \\
        (6) $\divergeDFG{e''}$ and $\reduceFGk{\kB}{\foreachN{D}}{\vbFG(e)}{e''}$ and $\kB < \kA$.

          \addSpace
        {\bf Subcase-Diverge-2:} $e' = v.m(v_1,...,v_{j-1},e'',...)$ for some $j$ where
        \\
        (7) $\divergeDFG{e''}$ and
        \\
        (8) $\reduceFGk{\kB}{\foreachN{D}}{\vbFG(e)}{v}$ and $\kB < \kA$ and
        \\
        (9) $\reduceFGk{k_j}{\foreachN{D}}{\vbFG(e_l)}{v_j}$ and $k_l < \kA$ for $l=1,...,j-1$ and
        \\
        (10) $\reduceFGk{\kC}{\foreachN{D}}{\vbFG(e_j)}{e''}$ and $\kC < \kA$.

          \addSpace
        {\bf Subcase-Diverge-3:}
        (11) $\reduceFGk{\kB}{\foreachN{D}}{\vbFG(e)}{v}$ and $\kB < \kA$ and
        \\
        (12) $\reduceFGk{k_i}{\foreachN{D}}{\vbFG(e_i)}{v_i}$ and $k_i < \kA$ and
        \\
         (13) $\foreachN{D} \turnsFG
                             \vbFG(e.m(\foreach{e_i}{n}))
                             \reduceSym^{\kB + \sum_i k_i}  v.m(\foreach{v_i}{n})
                             \reduceSym^{\kD} e'$
          where
          \\
          (14) $\kA = \kB + \sum_i k_i + \kD$.

      \addSpace
      Consider {\bf Subcase-Diverge-1}.

      (15) $\divergeMTL{\vbTL(\expT)}$ via (4), (6) and rule \Rule{red-rel-exp}.

      (16) $\divergeMTL{\vbTL(\mT{m}{t_S} \ \expT \ (\foreach{\expT_i}{n}))}$ via (15) and we are done here.

      \addSpace
      Consider {\bf Subcase-Diverge-2}.

      (17) $\reduceTLN{\vbMethTL}{\vbTL(\expT)}{\uT}$ for some $\uT$ via (4), (8) and rule \Rule{red-rel-exp}.

      (18) $\reduceTLN{\vbMethTL}{\vbTL(\expT_l)}{\uT_l}$ for some $\uT_l$ via (4), (9) and rule \Rule{red-rel-exp}
           for $l=1,..,j-1$.

      (19) $\divergeMTL{\vbTL(\expT_j)}$ via (4), (7), (1) and rule \Rule{red-rel-exp}.

           (20) $\divergeMTL{\mT{m}{t_S} \ \expT \ (\foreach{\expT_i}{n})}$ via (17), (18) and (19) and we are done here.

     \addSpace
      Consider {\bf Subcase-Diverge-3}.

       (20) $\redLRConfk{\kC+1}{t}
                     {v.m(\foreach{v_i}{n})}
                     {(\mT{m}{t_S} \ \uT) \ (\foreach{\uT_i}{n})}$
                     via the exact same reasoning steps that lead to (23) as found in the first subcase.
       (21) Recall $\kC = \minF(k-\kB,k - \sum_i k_i) - 1$.

       (22) Recall $\kC + 1 > \kA - (\kB - \sum_i k_i) = \kD$.

       (23) $\divergeMTL{\mT{m}{t_S} \ \expT \ (\foreach{\expT_i}{n})}$ via (13), (20), (22) and
            rule \Rule{red-rel-exp} and we are done here.

%% identical reasoning as in case of ``diverge''
\PANIC{
         \addSpace
        {\bf Subcase-Panic:}
        Suppose (5) $\reduceFGk{k'}{\foreachN{D}}{\vbFG(e.m(\foreach{e_i}{n}))}{e'}$ for some $e'$
        where $\kA < k$ and $\panicDFG{e'}$.

        We distinguish among the following three cases.
        Either
        (1) the expression on which the method call is performed panics or
        (2) one of the arguments panics or
        (3) reduction of the method call leads to some expression that panics.

          \addSpace
        {\bf Subcase-Panic-1:} $e' = e''.m(...)$ where
        \\
        (6) $\panicDFG{e''}$ and $\reduceFGk{\kB}{\foreachN{D}}{\vbFG(e)}{e''}$ and $\kB < \kA$.

          \addSpace
        {\bf Subcase-Panic-2:} $e' = v.m(v_1,...,v_{j-1},e'',...)$ for some $j$ where
        \\
        (7) $\panicDFG{e''}$ and
        \\
        (8) $\reduceFGk{\kB}{\foreachN{D}}{\vbFG(e)}{v}$ and $\kB < \kA$ and
        \\
        (9) $\reduceFGk{k_j}{\foreachN{D}}{\vbFG(e_l)}{v_j}$ and $k_l < \kA$ for $l=1,...,j-1$ and
        \\
        (10) $\reduceFGk{\kC}{\foreachN{D}}{\vbFG(e_j)}{e''}$ and $\kC < \kA$.

          \addSpace
        {\bf Subcase-Panic-3:}
        (11) $\reduceFGk{\kB}{\foreachN{D}}{\vbFG(e)}{v}$ and $\kB < \kA$ and
        \\
        (12) $\reduceFGk{k_i}{\foreachN{D}}{\vbFG(e_i)}{v_i}$ and $k_i < \kA$ and
        \\
         (13) $\foreachN{D} \turnsFG
                             \vbFG(e.m(\foreach{e_i}{n}))
                             \reduceSym^{\kB + \sum_i k_i}  v.m(\foreach{v_i}{n})
                             \reduceSym^{\kD} e'$
          where
          \\
          (14) $\kA = \kB + \sum_i k_i + \kD$.

      \addSpace
      Consider {\bf Subcase-Panic-1}.

      (15) $\panicMTL{\vbTL(\expT)}$ via (4), (6) and rule \Rule{red-rel-exp}.

      (16) $\panicMTL{\vbTL(\mT{m}{t_S} \ \expT \ (\foreach{\expT_i}{n}))}$ via (15) and we are done here.

      \addSpace
      Consider {\bf Subcase-Panic-2}.

      (17) $\reduceTLN{\vbMethTL}{\vbTL(\expT)}{\uT}$ for some $\uT$ via (4), (8) and rule \Rule{red-rel-exp}.

      (18) $\reduceTLN{\vbMethTL}{\vbTL(\expT_l)}{\uT_l}$ for some $\uT_l$ via (4), (9) and rule \Rule{red-rel-exp}
           for $l=1,..,j-1$.

      (19) $\panicMTL{\vbTL(\expT_j)}$ via (4), (7), (1) and rule \Rule{red-rel-exp}.

           (20) $\panicMTL{\mT{m}{t_S} \ \expT \ (\foreach{\expT_i}{n})}$ via (17), (18) and (19) and we are done here.

       \addSpace
       Consider {\bf Subcase-Panic-3}.

       (20) $\redLRConfk{\kC+1}{t}
                     {v.m(\foreach{v_i}{n})}
                     {(\mT{m}{t_S} \ \uT) \ (\foreach{\uT_i}{n})}$
                     via the exact same reasoning steps that lead to (23) as found in the first subcase.
       (21) Recall $\kC = \minF(k-\kB,k - \sum_i k_i) - 1$.

       (22) Recall $\kC + 1 > \kA - (\kB - \sum_i k_i) = \kD$.

       (23) $\panicMTL{\mT{m}{t_S} \ \expT \ (\foreach{\expT_i}{n})}$ via (13), (20), (22) and
            rule \Rule{red-rel-exp} and we are done here.

}
%% PANIC

  \noindent
  {\bf Case} \Rule{td-call-iface}:
  \begin{mathpar}
  \inferrule
            { \tdExpTrans{\pair{\foreachN{D}}{\fgEnv}}{e : t_I}{\expT}
             \\ (4) \ \TYPE\ t_I \ \INTERFACE\ \{ \foreach{S_i}{q} \} \in \foreachN{D}
             \\ (5) \ S_j = m(\foreach{x_i \ t_i}{n}) \ t\quad(\textrm{for some}~j \in [q])
            \\ \tdExpTrans{\pair{\foreachN{D}}{\fgEnv}}{e_i : t_i}{\expT_i}\quad\noteForall{i \in [n]}
          }
          { \tdExpTrans{\pair{\foreachN{D}}{\fgEnv}}
              {e.m(\foreach{e_i}{n}) : t}
              {\CASE\ \ \expT \ \OF\ \kT_{t_I} \ (\xTval, \foreach{\xT_i}{q}) \rightarrow \xT_j \ \xTval \ (\foreach{\expT_i}{n}) }
          }
  \end{mathpar}

  (6) $\redLRConfk{k}{t_I}{\vbFG(e)}{\vbTL(\expT)}$
     and (7) $\redLRConfk{k}{t_i}{\vbFG(e_i)}{\vbTL(\expT_i)}$ by induction.

  To establish (3) $\redLRConfk{k}{t}{\vbFG(e)}{\vbTL(\expT)}$
  we consider the subcases that arise in rule \Rule{red-rel-exp}.

  \addSpace
  {\bf Subcase-Terminate:}

  Suppose (8) $\reduceFGk{\kA}{\foreachN{D}}{\vbFG(e.m(\foreach{e_i}{n}))}{v'}$
  for some value $v'$ where (9) $\kA < k$.

  (10) $\reduceFGk{\kB}{\foreachN{D}}{\vbFG(e)}{v}$ for some $v$, $\kB$ where

  (11) $\kB < \kA$

  by observing the reduction (8) and the number of steps taken (9).

  (12) $\reduceFGk{k_i}{\foreachN{D}}{\vbFG(e_i)}{v_i}$ for some $v_i$, $k_i$ where

  (13) $\sum_i k_i < \kA$

  by observing the reduction (8) and the number of steps taken (9).

  (14) $\reduceTLN{\vbMethTL}{\vbTL(\expT)}{\uT}$ for some $\uT$ where

  (15) $\redVLRConfk{k-\kB}{t_I}{v}{\uT}$

  via reverse application of rule \Rule{red-rel-exp} on (6) where
  in the premise the left-hand side of the implication is satisfied via (10) and (11).

  (16) $\reduceTLN{\vbMethTL}{\vbTL(\expT_i)}{\uT_i}$ for some $\uT_i$ where

  (17) $\redVLRConfk{k-k_i}{t_i}{v_i}{\uT_i}$

  via reverse application of rule \Rule{red-rel-exp} on (7) where
  in the premise the left-hand side of the implication is satisfied via (12) and (13).

  (18) $\methodSpecifications{\foreachN{D}}{t_I} = \{ \foreach{S}{q} \}$

  via (4).

  (19) $\uTVal = \kT_{u_S} \ (\foreachN{\uTA})$ and

  (20) $\uT = \kT_{t_I} \ (\uTVal, \foreach{\yT_l}{q})$ for some $u_S$, $\foreachN{\uTA}$, $\foreach{\yT_l}{q}$ and

  (21) $\methodLookup{\foreachN{D}}{(m,u_S)} = \FUNC \ (x \ u_S) \ m (\foreach{x_i \ t_i}{n}) \ t \ \{ \RETURN\ e' \}$ and

  (22) $\redVLRConfk{k-\kB}{u_S}{v}{\uTVal}$ and

  (23) $\redLRConfk{k - \kB}
  {\ m (\foreach{x_i \ t_i}{n}) \ t}{\FUNC \ (x \ u_S) \ m (\foreach{x_i \ t_i}{n}) \ t \ \{ \RETURN\ e' \}}{\yT_j}$

  via reverse application of rule \Rule{red-rel-iface} on (15)
  where we assume (18) and (5).
    Index $k - \kB$ is the largest index that satisfies the logical relations
    in the premise of rule \Rule{red-rel-iface}.

  (24) Set $\kC = \minF(k-\kB,k - \sum_i k_i) - 1$ where

  (25) $0 \leq \kC < k$ and

  (26) $\kC < k- \kB$ and $\kC < k - k_i$

  via (9), (11) and (13).

  (27) $\redVLRConfk{\kC}{t_i}{v_i}{\uT_i}$

  via (17), (26) and the Monotonicity Lemma.

  (28) $\redVLRConfk{\kC}{u_S}{v}{\uTVal}$

  via (22), (26) and the Monotonicity Lemma.

   (29) $\redLRConfk{\kC}{t}
      {\Angle{\subst{x}{v},\foreach{\subst{x_i}{v_i}}{n}}}
      {\yT_j \ \uTVal \ (\foreach{\uT_i}{n})}$

      via reverse application of rule \Rule{red-rel-method} on (23)
      where in the premise the left-hand side of the implication is satisfied
      via (26), (27) and (28).

    (30) $\redLRConfk{\kC+1}{t}
      {v.m(\foreach{v_i}{n})}
      {\uTA_j \ \uTVal \ (\foreach{\uT_i}{n})}$

      via (29) and Lemma~\ref{le:source-step-preserve-equivalence}.

     (31) $\reduceFGk{\kA- (\kB + \sum_i k_i)}{\foreachN{D}}{v.m(\foreach{v_i}{n})}{v'}$

      via reduction steps (8), (10) and (12).

      (32) $\kA- (\kB + \sum_i k_i) < \kC + 1$

      via the following reasoning.

      \bda{llr}
        & \kC + 1
      \\ = & \minF(k - \kB,k - \sum_i k_i)  & \mbox{by definition}
      \\ = & k - \maxF(\kB, \sum_i k_i)     & \mbox{by $\minF/\maxF$ distributivity law}
      \\ \geq & k - (\kB + \sum_i k_i)      & \mbox{by $\maxF$ approximation}
      \\ > & \kA - (\kB + \sum_i k_i)       & \mbox{by $\kA < k$}
      \eda

      (33) $\reduceTLN{\vbMethTL}{\yT_j \ \uTVal \ (\foreach{\uT_i}{n})}{\uTA}$ for some $\uTA$ where

      (34) $\redLRConfk{\kC + 1 - (\kA - (\kB + \sum_i k_i))}{t}{v'}{\uTA}$

      via reverse application of rule \Rule{red-rel-exp} on (30)
      where in the premise the left-hand side of the implication is satisfied via (31) and (32).

      (35) $k - \kA \leq \kC + 1 - (\kA - (\kB + \sum_i k_i))$

      via the following reasoning.

      \bda{llr}
      & \kC + 1 - (\kA - (\kB + \sum_i k_i))
      \\ = & \minF(k - \kB, k - \sum_i k_i) - (\kA - (\kB + \sum_i k_i))   & \mbox{by definition}
      \\ = & k - \maxF(\kB, \sum_i k_i) - (\kA  - (\kB + \sum_i k_i))      & \mbox{by $\minF/\maxF$ distributivity law}
      \\ = & (k - \kA) + (\kB + \sum_i k_i) - \maxF(\kB, \sum_i k_i)       & \mbox{by equivalence}
      \\ \geq &  k - \kA                                                  & \mbox{by approximation}
      \eda

      (36) $\redLRConfk{k-\kA}{t}{v'}{\uTA}$

      via (34), (35) and the Monotonicity Lemma.

      (37) $\reduceTLN{\vbMethTL}{\vbTL(\CASE\ \ \expT \ \OF\ \kT_{t_I} \ (\xTval, \foreach{\xT_i}{q}) \rightarrow \xT_j \ \xTval \ (\foreach{\expT_i}{n}))}{\uTA}$

      via reduction steps (14), (20) and (33).

      Via (36) and (37) we can establish this subcase in statement (3).

      \addSpace
          {\bf Subcase-Diverge} and {\bf Subcase-Panic} are left out as the reasoning
          is very close to the reasoning for case \Rule{td-call-struct}.

  \noindent
  {\bf Case} \Rule{td-sub}:
  \begin{mathpar}
  \inferrule
            {\tdExpTrans{\pair{\foreachN{D}}{\fgEnv}}{e : t}{\expT_2}
          \\ \tdUpcast{\foreachN{D}}{\subtypeOf{t}{u}}{\expT_1}
          }
            { \tdExpTrans{\pair{\foreachN{D}}{\fgEnv}}{e : u}{\expT_1 \ \expT_2} }
  \end{mathpar}

  By induction we obtain that
  (4) $\redLRConfk{k}{t}{\vbFG(e)}{\vbTL(\expT_2)}$.
  From (3), (4) and Lemma~\ref{le:upcast-red-equiv}  we obtain that
  $\redLRConfk{k}{u}{\vbFG(e)}{\expT_1 \ \vbTL(\expT_2)}$.

  We have that $\vbTL(\expT_1) = \expT_1$ and thus we are done for this case.

  \noindent
  {\bf Case} \Rule{td-assert}:
  \begin{mathpar}
  \inferrule
            {\tdExpTrans{\pair{\foreachN{D}}{\fgEnv}}{e : u}{\expT_2}
          \\ \tdDowncast{\foreachN{D}}{\assertOf{u}{t}}{\expT_1}
          }
          { \tdExpTrans{\pair{\foreachN{D}}{\fgEnv}}{e.(t) : t}{\expT_1 \ \expT_2} }
  \end{mathpar}

    By induction we obtain that
  (4) $\redLRConfk{k}{u}{\vbFG(e)}{\vbTL(\expT_2)}$.
  From (3), (4) and Lemma~\ref{le:downcast-red-equiv}  we obtain that
  $\redLRConfk{k}{t}{\vbFG(e).(t)}{\expT_1 \ \vbTL(\expT_2)}$.

  We have that $\vbTL(\expT_1) = \expT_1$ and thus we are done for this case.
  \qed
\end{proof}

\end{document}